\documentclass[10pt]{amsart}
\usepackage{amsmath, amsfonts, amsthm, amssymb}
\usepackage[a4paper, width=16cm, top=30mm, bottom=30mm]{geometry}
\usepackage{mathrsfs}
\usepackage{verbatim}
\usepackage{hyperref}

\usepackage{color}
\usepackage{bbm}
\usepackage{ulem}

\usepackage{graphicx}
\usepackage{float}
\usepackage{enumerate}
\numberwithin{equation}{section}
\newtheorem{theorem}{Theorem}[section]
\newtheorem{corollary}[theorem]{Corollary}
\newtheorem{lemma}[theorem]{Lemma}
\newtheorem{proposition}[theorem]{Proposition}

\theoremstyle{definition}
\newtheorem{definition}[theorem]{Definition}
\newtheorem{remark}[theorem]{Remark}
\newtheorem{assumption}[theorem]{Assumption}

\newtheorem{model}[theorem]{Model}

\newcommand{\ind}{1\hspace{-2.1mm}{1}}

\newcommand{\RR}{\mathbb{R}}
\newcommand{\PP}{\mathbb{P}}

\newcommand{\EE}{\mathbb{E}}

\newcommand{\Cc}{\mathcal{C}}

\newcommand{\Ii}{\mathcal{I}}

\newcommand{\Tt}{\mathcal{T}}

\newcommand{\D}{\mathrm{d}}

\newcommand{\xx}{\mathrm{x}}
\newcommand{\yy}{\mathrm{y}}
\newcommand{\zz}{\mathrm{z}}
\newcommand{\E}{\mathrm{e}}

\newcommand{\Ff}{\mathscr{F}}

\newcommand{\eps}{\varepsilon}

\newcommand{\CCT}{\Cc[T,T+\Delta]}
\newcommand{\VgT}{V^{g,T}}
\newcommand{\HgT}{\mathscr{H}^{g,T}}
\newcommand{\IgT}{\Ii^{g,T}}
\newcommand{\be}{\begin{equation}}
\newcommand{\ee}{\end{equation}}
\DeclareMathOperator*{\argmin}{\arg\!\min}

\begin{document}

\title{Asymptotics for volatility derivatives in multi-factor rough volatility models}

\author{Chloe Lacombe}
\address{Department of Mathematics, Imperial College London}
\email{chloe.lacombe14@imperial.ac.uk}

\author{Aitor Muguruza}
\address{Department of Mathematics, Imperial College London and Synergis}
\email{aitor.muguruza-gonzalez15@imperial.ac.uk}

\author{Henry Stone}
\address{Department of Mathematics, Imperial College London}
\email{henry.stone15@imperial.ac.uk}
\thanks{The authors are grateful to Antoine Jacquier, Mikko Pakkanen and Ryan McCrickerd for stimulating discussions.  AM and HS thank the EPSRC CDT in
Financial Computing and Analytics for financial support}
\date{\today}
\maketitle
\begin{abstract}
We study the small-time implied volatility smile for Realised Variance options,
and investigate the effect of correlation in multi-factor models on the linearity of the smile.
We also develop an approximation scheme for the Realised Variance density, allowing fast and accurate pricing of Volatility Swaps.
Additionally, we establish small-noise asymptotic behaviour of a general class of VIX options in the large strike regime.
\end{abstract}

\section{Introduction}
Following the works by Al\`os, Le\'on and Vives \cite{AlosLeon}, Gatheral, Jaisson, and Rosenbaum \cite{GJR18} and Bayer, Friz and Gatheral~\cite{BFG16}, rough volatility is becoming a new breed in financial modelling by generalising Bergomi's `second generation' stochastic volatility models to a non-Markovian setting. The most basic form of (lognormal) rough volatility model is the so-called rough Bergomi model introduced in~\cite{BFG16}.
Gassiat \cite{Gas18} recently proved that such a model (under certain correlation regimes) generates true martingales for the spot process. The lack of Markovianity imposes numerous fundamental theoretical questions and practical challenges in order to make rough volatility usable in an industrial environment.
On the theoretical side, Jacquier, Pakkanen, and Stone \cite{JPS18} prove a large deviations principle for a rescaled version of the log stock price process.
In this same direction, Bayer, Friz, Gulisashvili, Horvath and Stemper \cite{BFGHS},
Forde and Zhang \cite{FZ17},
Horvath, Jacquier and Lacombe \cite{HJL18} and most recently Friz, Gassiat and Pigato \cite{FGP18} (to name a few) prove large deviations principles for a wide range of rough volatility models. On the practical side, competitive simulation methods are developed in Bennedsen, Lunde and Pakkanen \cite{BLP15}, Horvath, Jacquier and Muguruza \cite{HJM17} and McCrickerd and Pakkanen \cite{MP18}. Moreover, recent developments by Stone \cite{Sto18} and Horvath, Muguruza and Tomas \cite{HMT19} allow the use of neural networks for calibration; their calibration schemes are considerably faster and more accurate than existing methods for rough volatility models.

Crucially, the lack of a pricing PDE imposes a fundamental constraint on the comprehension and interpretation of rough volatility models driven, by Volterra-like Gaussian processes. The only current exception in the rough volatility literature is the rough Heston model, developed by El Euch and Rosenbaum \cite{ER18,ER19}, which allows a better understanding through the fractional PDE derived in \cite{ER19}. Nevertheless, in this work our attention is turned to the class of models for which such a pricing PDE is unknown, and hence further theoretical results are required.

 Perhaps, options on volatility itself are the most natural object to first analyse within the class of rough volatility models. In this direction,
Jacquier, Martini, and Muguruza \cite{JMM18} provide algorithms for pricing VIX options and futures. Horvath, Jacquier and Tankov \cite{HJT} further study VIX smiles in the presence of stochastic volatility of volatility combined with rough volatility.
Nevertheless, the precise effect of model parameters (with particular interest in the Hurst parameter effect) on implied volatility smiles for VIX (or volatility derivatives in general) has not been studied until very recently in Al\`os, Garc\'ia-Lorite and Muguruza \cite{AGM18}.

The main focus of the paper is to derive the small-time behaviour of the realised variance process of the rough Bergomi model, as well as related but more complicated multi-factor rough volatility models, together with the small-time behaviour of options on realised variance.
These results, which are interesting from a theoretical perspective, have practical applicability to the quantitative finance industry as they allow practitioners to better understand the Realised Variance smile, as well as
the effect of correlation on the smile's linearity (or possibly convexity).
To the best of our knowledge, this is the first paper to study the small-time behaviour of options on realised variance.
An additional major contribution of the paper is the numerical scheme used to compute the implied volatility smiles, using an accurate approximation of the rate function from a large deviations principle. In general rate functions are highly non-trivial to compute; our method is simple, intuitive, and accurate.
The numerical methods are publicly available on \href{https://github.com/amuguruza/LDP-VolOptions}{GitHub: LDP-VolOptions}.

Volatility options are becoming increasingly popular in the financial industry. For instance, VIX options' liquidity has consistently increased since its creation by the Chicago Board of Exchange (CBOE). One of the main popularity drivers is that volatility tends to be negatively correlated with the underlying dynamics, making it desirable for portfolio diversification.
Due to the appealing nature of volatility options, their modelling has attracted the attention of many academics such as Carr, Geman, Madan and Yor \cite{CGMY05}, Carr and Lee \cite {CL09} to name a few.

For a log stock price process $X$ defined as
$ X_t  =  - \frac{1}{2} \int_0^t v_s \D s + \int_0^t \sqrt{ v_s } \D B_s,  X_0=0,
$
where $B$ is standard Brownian motion, we denote the quadratic variation of $X$ at time $t$ by $\langle X \rangle_t$.
Then, the core object to analyse in this setting is the realised variance option with payoff
\begin{equation}\label{eq:RVOption}
\left(\frac{1}{T}\int_0^T \D \langle X \rangle_s -K\right)^+,
\end{equation}
which in turn defines the risk neutral density of the realised variance.  
In this work, we analyse the short time behaviour of the implied volatility given by \eqref{eq:RVOption} for a number of (rough) stochastic volatility models by means of large deviation techniques.  
We specifically focus on the construction of correlated factors and their effect on the distribution of the realised variance.
We find our results consistent with that of Al\`os, Garc\'ia-Lorite and Muguruza \cite{AGM18}, which also help us characterise in close-form the implied volatility around the money. Moreover, we also obtain some asymptotic results for VIX options.

While implied volatilities for options on equities are typically convex functions of log-moneyness, giving them their ``smile'' moniker, implied volatility smiles for options on realised variance tend to be linear.
Options on integrated variance are OTC products, and so their implied volatility smiles are not publicly available. VIX smiles are, however, and provide a good proxy for integrated variance smiles; see Figure \ref{VIX plots} below for evidence of their linearity.
The data also indicates both a power-law term structure ATM and its skew.

\begin{figure}[h!]
\hspace*{-20mm}
\includegraphics[scale=0.45]{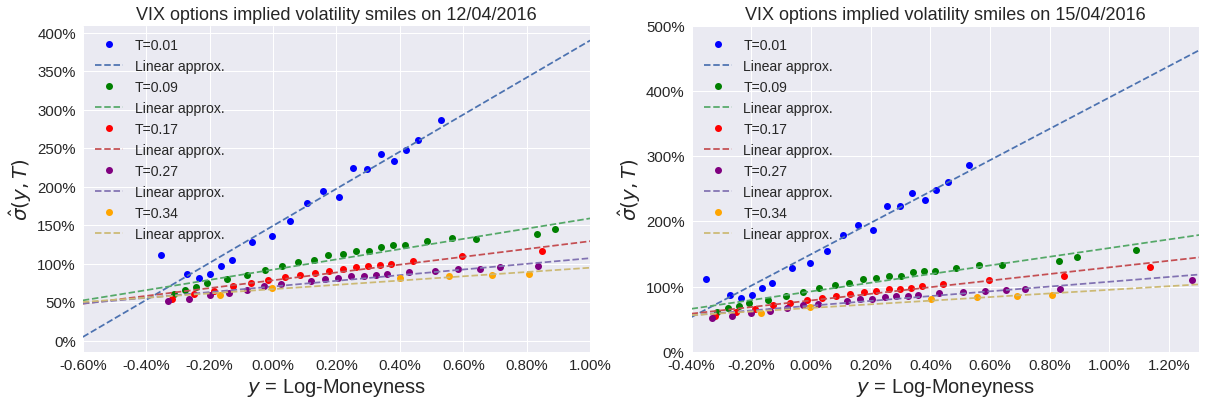}
\caption{ Implied volatility smiles for Call options on VIX for small maturities, close to the money.
Data provided by OptionMetrics.
}
\label{VIX plots}
\end{figure}

In spite of most of the literature agreeing on the fact that more than a single factor is needed to model volatility (see Bergomi's \cite{Ber16} two-factor model, Avellaneda and Papanicolaou \cite{AP18} or Horvath, Jacquier and Tankov \cite{HJT} for instance), there is no in-depth analysis on how to construct these (correlated) factors, nor the effect of correlation on the price of volatility derivatives and their corresponding implied volatility smiles. Our aim is to understand multi-factor models and analyse the effect of factors in implied volatility smiles.
This paper, to the best of our knowledge is the first to address such questions, which are of great interest to practitioners in the quantitative finance industry; it is also the first to provide a rigorous mathematical analysis of the small-time behaviour of options on integrated variance in rough volatility models.

The structure of the paper is as follows. Section \ref{sec:models} introduces the models, the rough Bergomi model and two closely related processes,
whose small-time realised variance behaviour we study; the main results are given in Section \ref{section: options on int variance}.
Section \ref{sec: numerics} is dedicated to numerical examples of the results attained in Section \ref{section: options on int variance}.
In Section \ref{section: VIX options} we introduce a general variance process, which includes the rough Bergomi model for a specific choice of kernel, and briefly investigate the small-noise behaviour of VIX options in this general setting.
Motivated by the numerical examples in Section \ref{sec: numerics}, we propose a simple and very feasible approximation for the density of the realised variance for the mixed rough Bergomi model (see \eqref{eq:mixedvarianceProcess}) in Appendix \ref{app:linear approximation}.
The proofs of the main results are given in Appendix \ref{app: proof of main results}; the details of the numerics are given in Appendix \ref{appendix: numerics}.

\textbf{Notations:}
Let $\RR_+ := [0,+\infty)$ and $\RR^*_+ := (0,+\infty)$.
For some index set~$\Tt \subseteq \RR_+$, the notation $L^2(\Tt)$ denotes the space of real-valued square integrable functions on~$\Tt$,
and $ \Cc(\Tt, \RR^d)$ the space of $\RR^d$-valued continuous functions on~$\Tt$.
 $\mathcal{E}$ denotes the Wick stochastic exponential.

\section{A showcase of rough volatility models}\label{sec:models}
In this section we introduce the models that will be considered in the forthcoming computations.
We shall always work on a given filtered probability space $(\Omega,(\Ff_t)_{t\geq 0},\mathbb{P})$. For notational convenience, we introduce
\begin{equation}\label{eq:SDEZ}
Z_t := \int_0^t K_\alpha(s,t)\D W_s,
\qquad\text{for any }t \in\Tt,
\end{equation}
where $\alpha \in \left(-\frac{1}{2},0\right)$, $W$ a standard Brownian motion,
and where the kernel
$K_{\alpha}:\RR_+\times\RR_+ \to \RR_+$ reads
\begin{equation}\label{eq:K}
K_{\alpha}(s,t) := \eta \sqrt{2\alpha + 1}(t-s)^{\alpha},
\qquad \text{for all } 0\leq s<t,
\end{equation}
for some strictly positive constant~$\eta$.
Note that, for any $t\geq 0$, the map $s\mapsto K_\alpha(s,t)$ belongs to~$L^2(\Tt)$,
so that the stochastic integral~\eqref{eq:SDEZ} is well defined. We also define an analogous multi-dimensional version of \eqref{eq:SDEZ} by

\begin{equation}\label{eq:nDimensionalSDEZ} \mathcal{Z}_t:=\left(
\int_0^t K_\alpha(s,t)\D W^1_s, ..., \int_0^t K_\alpha(s,t)\D W^m_s
\right)
:=
\left(
\mathcal{Z}^1_t, ..., \mathcal{Z}_t^m
\right), \qquad\text{for any }t \in\Tt,
\end{equation}
where $W^1,...,W^m$ are independent Brownian motions.

\begin{model}[Rough Bergomi]
The rough Bergomi model, where $X$ is the log stock price process and $v$ is the instantaneous variance process, is then defined (see \cite{BFG16}) as
\begin{equation}\label{rough Bergomi}
\begin{array}{rll}
X_t & = \displaystyle - \frac{1}{2} \int_0^t v_s \D s + \int_0^t \sqrt{ v_s } \D B_s,
 \quad &  X_0 = 0 , \\
v_t &= \displaystyle v_0
\exp \left( Z_t -\frac{\eta^2}{2}t^{2 \alpha +1}  \right),
\quad & v_0 > 0,
\end{array}
\end{equation}
where the Brownian motion~$B$ is defined as $B := \rho W + \sqrt{1-\rho^2}W^\perp$ for $\rho \in [-1,1]$
and some standard Brownian motion~$W^\perp$ independent of $W$.
\end{model}

\begin{remark}\label{remark: H and alpha}
The process $\log v$ has a modification whose sample paths are almost surely locally $\gamma$-H\"older continuous,
for all $\gamma \in \left(0, \alpha + \frac{1}{2} \right)$ \cite[Proposition 2.2]{JPS18}.
For rough volatility models involving a fractional Brownian motion, the sample path regularity of the log volatility process is referred to in terms of the Hurst parameter $H$; recall that the fractional Brownian motion has sample paths that are $\gamma$-H\"older continuous
for any $\gamma \in (0,H)$~\cite[Theorem 1.6.1]{BHOZ08}.
By identification, therefore, we have that $\alpha= H - 1/2$.
\end{remark}

\begin{model}[Mixed rough Bergomi]
The mixed rough Bergomi model is given in terms of log stock price process $X$ and instantaneous variance process $v^{(\gamma,\nu)}$ as
 \begin{equation}\label{eq:mixedvarianceProcess}
\begin{array}{rll}
X_t & = \displaystyle - \frac{1}{2} \int_0^t v_s^{(\gamma,\nu)} \D s + \int_0^t \sqrt{ v_s^{(\gamma,\nu)} } \D B_s,
 \quad &  X_0 = 0 , \\v_t^{(\gamma,\nu)}
 &= v_0 \sum_{i=1}^n \gamma_i \exp\left( \frac{\nu_i}{\eta}Z_t - \frac{\nu_i^2}{2}t^{2\alpha +1}\right) ,\quad & v_0>0
 \end{array}
\end{equation}
where  $\gamma:=(\gamma_1,...,\gamma_n)\in[0,1]^n$ such that $\sum_{i=1}^n \gamma_i =1$ and $\nu:=(\nu_1,...,\nu_n)\in\mathbb{R}^n$, such that $ 0<\nu_1<...<\nu_n$.
\end{model}
The above modification of the rough Bergomi model, inspired by Bergomi \cite{Bergomi3}, allows a bigger slope (hence bigger skew) on the implied volatility of variance/volatility options to be created, whilst maintaining a tractable instantaneous variance form. This will be made precise in Section \ref{section: numerics mixed case}.
\begin{model}[Mixed multi-factor rough Bergomi]
The mixed rough Bergomi model is given in terms of log stock price process $X$ and instantaneous variance process $v^{(\gamma,\nu,\Sigma)}$ as
\begin{equation}\label{eq:mixedvarianceMultiFactorProcess}
\begin{array}{rll}
X_t & = \displaystyle - \frac{1}{2} \int_0^t v_s^{(\gamma,\nu,\Sigma)} \D s + \int_0^t \sqrt{ v_s^{(\gamma,\nu,\Sigma)} } \D B_s,
 \quad &  X_0 = 0 , \\v_t^{(\gamma,\nu,\Sigma)}&= v_0 \sum_{i=1}^{n} \gamma_i \mathcal{E}\left(\frac{\nu^i}{\eta}\cdot\mathrm{L}_i \mathcal{Z}_t\right) ,\quad & v_0>0,
 \end{array}
\end{equation}
where $\gamma:=(\gamma_1,...,\gamma_n)\in[0,1]^n$ such that $\sum_{i=1}^n \gamma_i =1$. The vector $\nu^i=(\nu^i_1,...,\nu^i_m)\in\mathbb{R}^m$
satisfies $0<\nu^i_1<...<\nu^i_m$ for all $i\in\{1,...,n\}$.
In addition, $\mathrm{L}_i\in\mathbb{R}^{m\times  m}$ is a lower triangular matrix such that $\mathrm{L}_i \mathrm{L}^T_i =: \Sigma_i$ is a positive definite matrix for all $i\in\{1,...,n\}$, denoting the covariance matrix.
\end{model}

For all results involving models \eqref{rough Bergomi}, \eqref{eq:mixedvarianceProcess}, and \eqref{eq:mixedvarianceMultiFactorProcess} we fix $\Tt=[0,1]$; minor adjustments to the proofs yield analogous results for more general $\Tt$.
We additionally define $\beta:=2\alpha+1\in(0,1)$  for notational convenience.

\begin{remark}
In models \eqref{rough Bergomi}-\eqref{eq:mixedvarianceMultiFactorProcess} we have considered a flat or constant initial forward variance curve $v_0>0$. However, our framework can be easily extended to functional forms $v_0(\cdot): \Tt \mapsto \RR_+$ via the Contraction Principle, see Appendix \ref{appendix: contraction princ}, as long as the mapping is continuous.  
\end{remark}
\begin{remark}
The reader may already have realised that the mixed multi-factor rough Bergomi defined in \eqref{eq:mixedvarianceMultiFactorProcess} is indeed general enough to cover both \eqref{rough Bergomi} and \eqref{eq:mixedvarianceProcess}.
However, we provide our theoretical results in an orderly fashion starting from  \eqref{rough Bergomi} and finishing with \eqref{eq:mixedvarianceMultiFactorProcess}, which we find the most natural way to increase the complexity of the model.
\end{remark}

In place of $K_\alpha$ in \eqref{eq:SDEZ}, one may also consider more general kernels of the form
$$G(t,s):=K_{\alpha}(t,s) L(t-s)$$
where $L\in\mathcal{C}(0,\infty)$ is a measurable function such that the stochastic integral is well defined, $L(\cdot) $ is decreasing and $L(\cdot)$ continuous at $0$. We note that under such conditions, $L(\cdot)$ is also slowly varying at zero, i.e. $\lim_{x\to 0} \frac{L(tx)}{L(x)}=1$ for any $t>0$. 
Such kernels are naturally related to the class of Truncated Brownian Semistationary ($\mathcal{TBSS}$) processes introduced by Barndorff-Nielsen and Schmiegel \cite{turbulence}. Examples include the Gamma and Power-law kernels:
\begin{eqnarray*}
L_{\text{Gamma}}(t-s)&=\exp(-\kappa(t-s)), &\kappa>0\\
L_{\text{Power}}(t-s)&=(1+t-s)^{\zeta-\alpha}, &\zeta<-1.
\end{eqnarray*}
The following result gives the exponential equivalence between the sequences of the rescaled stochastic integrals of
$K_\alpha$ and $G$, thus it is completely justified to only consider the case $K_{\alpha}$, without any loss of generality.

\begin{proposition}\label{pp:exp_equi L}
The sequences of processes
$\left(\eps^{\beta / 2 }\int_0^\cdot K_\alpha (\cdot,s) L(\eps(\cdot-s))\D W_s\right)_{\eps>0}$ and $\left(\eps^{\beta / 2}Z_\cdot\right)_{\eps>0}$ are exponentially equivalent.
\end{proposition}
\begin{proof}
As $L$ is a slowly varying function, the so-called Potter bounds \cite[Theorem 1.5.6, page 25]{Bingham} hold on the interval $(0,1]$: indeed, for all $\xi >0$, there exist constants $0< \underline{C}^\xi \le \overline{C}^\xi $ such that
$$
\underline{C}^\xi (\eps x)^\xi < L(\eps x) < \overline{C}^\xi (\eps x)^\xi, \quad \text{for all } (\eps x)\in (0,1].
$$
In particular, for $\eps >0$ such that $(\eps x)\in[0,1]$,
$L(\eps x) -1 \le K_\xi (\eps x)^\xi$ where $K_\xi < \infty$ as $\xi >0$.
Thus, for all $\delta >0$,
$$
\begin{aligned}\displaystyle
\PP\left(\left\| \eps^{\beta / 2} \int_0^\cdot K_\alpha (\cdot, s)\left[L(\eps(\cdot -s))-1 \right] \D W_s \right\|_{\infty} \!\!\!\!  > \delta \right)
&= \PP \left( |\mathcal{N}(0,1)| > \frac{\delta}{\eps^{\beta / 2} \left\| V_\eps \right\|_{\infty}} \right) \\
&= \sqrt{\frac{2}{\pi}} \frac{\eps^{\beta / 2} \left\| V_\eps \right\|_{\infty}}{\delta} \exp \left( \frac{-\delta^2}{2\eps^{\beta} \left\| V_\eps \right\|^2_{\infty}}\right)(1+\mathcal{O}(\eps^{\beta})),
\end{aligned}
$$
where the final equality follows by using the asymptotic expansion of the Gaussian density near infinity~\cite[Formula (26.2.12)]{AS72}, and $V^2_\eps (t) := \int_0^t K^2_\alpha (t,s) \left[ L(\eps (t-s)) - 1\right]^2 \D s$,
Then,
$$
V^2_\eps (t)
\le K^2_\xi \eps^{2\xi} \eta^2 (2\alpha +1) \int_0^t (t-s)^{2(\alpha + \xi)} \D s
=  \left(K_\xi \eta\right)^2 \frac{2\alpha +1}{2(\alpha + \xi)+1} \eps^{2\xi} t^{2(\alpha + \xi) +1},
$$
so that
$\lim_{\eps \downarrow 0} \left\| V^2_\eps \right\|_\infty =0$, as well as $\lim_{\eps \downarrow 0} \left\| V_\eps \right\|_\infty =0$.
Therefore
\begin{align*}
\eps^\beta \log \PP \left( \left\| \eps^{\beta / 2} \int_0^\cdot K_\alpha(\cdot, s) \left[ L(\eps(\cdot -s)) -1 \right] \D W_s \right\|_\infty > \delta \right)
& \le
\frac{\eps^\beta}{2} \log \left(\frac{2}{\pi}\right)
+ \eps^\beta \log \left( \frac{\eps^{\beta / 2} \left\| V_\eps \right\|_\infty}{\delta}\right) \\
&- \frac{\delta^2}{2\left\|V_\eps \right\|^2_\infty}
+ \eps^\beta \log \left(1+\mathcal{O})\eps^\beta \right).
\end{align*}
As $\eps$ tends to zero, the first  and second terms in the above inequality tend to zero (recall that $0<\beta<1$),
and the third term tends to $-\infty$.
Hence for all $\delta >0$,
$$
\limsup_{\eps \downarrow 0} \eps^\beta \log \PP \left( \left\| \eps^{\beta / 2} \int_0^\cdot K_\alpha(\cdot, s) \left[ L(\eps(\cdot -s)) -1 \right] \D W_s \right\|_\infty > \delta \right) =- \infty,
$$
and thus the two processes are exponentially equivalent \cite[Definition 4.2.10]{DZ10}; see Appendix \ref{appendix: contraction princ}.
\end{proof}

\begin{corollary} The sequences of processes
$(\eps^{\beta/2} \int_0^t \E^{-\kappa_i \eps(t-s)} K_\alpha(s,t)\D W^i)_{\eps >0}$ and~$(\eps^{\beta/2} \int_0^t K_\alpha(s,t)\D W^i)_{\eps>0}$ are exponentially equivalent for $i=1,...,m$, where each $\kappa_i>0$.
\end{corollary}

\begin{proof}
The proof is similar to the proof of Proposition~\ref{pp:exp_equi L}. The variance in the asymptotic expansion of the Gaussian density near infinity~\cite[Formula (26.2.12)]{AS72} is defined as
$$
V^2_\eps(t) := \int_0^t \left[ \E^{\kappa_j \eps (t -s)} -1 \right]^2 K^2_\alpha (s,t) \D s =
\eta^2 \beta
\left[ \frac{t^\beta}{\beta} -2\int_0^t \E^{-\kappa_j \eps (t-s)} (t-s)^{2\alpha} \D s
+ \int_0^t \E^{-2\kappa_j \eps (t-s)} (t-s)^{2\alpha} \D s\right],
$$
such that
$$
0
<\eps^\beta V^2_\eps
\le \eta^2 \beta \eps^\beta
\left[ \frac{t^\beta}{\beta}
+ \int_0^t  \E^{-2\kappa_j \eps (t-s)} (t-s)^{2\alpha} \D s \right]
\le 2 \eta^2 (\eps t)^\beta ,
$$
and therefore $\lim_{\eps \downarrow 0} V^2_\eps = 0$ and $\lim_{\eps \downarrow 0} \eps^\beta \left\|V_\eps \right\|^2_\infty = 0$. Then, for all $\delta >0$,
$$
\limsup_{\eps \downarrow 0}
\eps^\beta \log \mathbb{P} \left( \left\|
\eps^{\beta/2} \int_0^{\cdot} K_{\alpha}(s,\cdot) \left[ \exp (-\kappa_j \eps (\cdot-s))-1\right]\D W^j_s\right\|_{\infty} > \delta \right)= -\infty,
$$
ans thus the two processes are exponentially equivalent \cite[Definition 4.2.10]{DZ10}.
\end{proof}

\section{Small-time results for options on integrated variance}\label{section: options on int variance}
We start our theoretical analysis by considering options on realised variance, which we also refer to as integrated variance and RV interchangeably.
We recall that volatility is not directly observable, nor a tradeable asset. Options on realised variance, however, exist and are traded as OTC products.
Below are two examples of the payoff structure of such products:
\begin{equation}\label{integrated vol payoffs}
(i) (RV(v)(T)-K)^+, \quad (ii) (\sqrt{RV(v)(T)}-K)^+, \quad \textrm{where } T,K \ge 0.
\end{equation}
where we define the following $\Cc(\Tt)$ operator
\begin{equation}\label{IV operator}
RV(f)(\cdot):f \mapsto \frac{1}{\cdot} \int_0^\cdot f(s) \D s,
\quad
RV(f)(0):=f(0),
\end{equation}

and $v$ represents the instantaneous variance in a given stochastic volatility model.
Note that $RV(v)(0)~=~v_0$.
\begin{remark}As shown by Neuberger \cite{Neu94}, we may rewrite the variance swap in terms of the log contract as
\begin{equation}\label{eq:log-contract}
\mathbb{E}[RV(v)(T)]=\mathbb{E}\left[\frac{1}{T}\int_0^T v_s \D s\right]=\mathbb{E}\left[-2\frac{X_T}{T}\right]
\end{equation}
where $\mathbb{E}[\cdot]$ is taken under the risk-neutral measure and $S=\exp(X)$ is a risk-neutral martingale (assuming interest rates and dividends to be null). Therefore, the risk neutral pricing of $RV(v)(T)$ or options on it is fully justified by \eqref{eq:log-contract}.
\end{remark}

\subsection{Small-time results for the rough Bergomi model}

\begin{proposition}\label{RKHS for v}
The set $\mathscr{H}^{K_{\alpha}}:=\{\int_0^{\cdot} K_{\alpha}(s,\cdot)f(s) \D s : f\in L^2(\Tt) \}  $ defines the reproducing kernel Hilbert space for $Z$ with inner product
$\langle \int_0^{\cdot} K_{\alpha}(s,\cdot)f_1(s)\D s, \int_0^{\cdot} K_{\alpha}(s,\cdot)f_2(s)\D s  \rangle_{\mathscr{H}^{K_{\alpha}}}  
=
\langle f_1, f_2 \rangle_{L^2(\Tt)}
$.
\end{proposition}

\begin{proof}
See \cite[Theorem 3.1]{JPS18}.
\end{proof}

Before stating Theorem \ref{LDP for vol process}, we define the following function
$\Lambda^{Z}: \Cc(\Tt) \to\RR_+$ as
$
\Lambda^{Z}(\xx):= \frac{1}{2}\|\xx\|_{\mathscr{H}^{K_\alpha}}^2
$,
and if $\xx \notin \mathscr{H}^{K_{\alpha}}$ then $\Lambda^{Z}(\xx) = + \infty $.

\begin{theorem}\label{LDP for vol process}

The variance process $(v_\eps)_{\eps >0}$ satisfies a large deviations principle on $\Cc(\Tt)$ as $\eps$ tends to zero, with speed $\eps^{-\beta}$ and rate function
$\Lambda^v(x):=
\Lambda^Z\left( \log\left(\frac{x}{v_0} \right)  \right)$, where $\Lambda^v(v_0)=0$ and $x\in\mathcal{C}(\mathcal{T})$.
\end{theorem}
\begin{proof}
To ease the flow of the paper the proof is postponed to Appendix \ref{3.3}.
\end{proof}

\begin{corollary}\label{integrated vol LDP}
The integrated variance process $\left(RV(v)(t)\right)_{t \in \Tt }$ satisfies a large deviations principle on $\RR^*_+$ as $t$ tends to zero, with speed $t^{-\beta}$ and rate function
$\hat{\Lambda}^v$ defined as $\hat{\Lambda}^v(\yy) := \inf \left\{ \Lambda^v(\xx) : \yy = RV(\xx)(1) \right\}$, where $\hat{\Lambda}^v(v_0)=0$.
\end{corollary}

\begin{proof}
As proved in Theorem~\ref{LDP for vol process}, the process $(v_\eps)_{\eps>0}$ satisfies a pathwise large deviations principle on $\Cc(\Tt)$ as $\eps$ tends to zero.
For small perturbations $\delta^v \in \Cc (\Tt)$, we have
$$
\left\|RV(v+\delta^v)(t)-RV(\delta^v)(t)\right\|_\infty
\le \sup_{t\in \Tt} \frac{1}{t} \left| \int_0^t \delta^v(s) \D s \right|
\le M,
$$
where $M= \sup_{t\in \Tt} |\delta^v(t)|$, which is finite as $\delta^v \in \Cc (\Tt)$.
Clearly $M$ tends to zero as $ \delta^v$ tend to zero,
and hence the operator $RV$ is continuous with respect to the sup norm on $\Cc (\Tt)$.
Therefore we can apply the Contraction Principle (see Appendix \ref{appendix: contraction princ}), and consequently the integrated variance process $RV(v_\eps)$ satisfies a large deviations principle on $\Cc(\Tt)$ as $\eps$ tends to zero.
Clearly $RV(v_\eps)(t)=RV(v)(\eps t)$, for all $t\in\Tt$, and so setting $t=1$ and mapping $\eps$ to $t$ then yields the result.
By definition, $\hat{\Lambda}^v(\yy) := \inf \left\{ \Lambda^v(\xx) : \yy = RV(\xx)(1) \right\}$. If we choose $\xx\equiv v_0$ then clearly $v_0=RV(\xx)(1)$, and $\Lambda^v(\xx)=0$. Since $\Lambda^v$ is a norm, it is a non-negative function and therefore  $\hat{\Lambda}^v(v_0)=0$. This concludes the proof.
\end{proof}

\begin{remark}\label{remark: integrated variance remark}
Corollary \ref{integrated vol LDP} can be applied to a large number of existing results on large deviations for (rough) variance processes to get a large deviations result for the integrated (rough) variance process; for example Forde and Zhang \cite{FZ17} and Horvath, Jacquier, and Lacombe \cite{HJL18}.
\end{remark}

\begin{corollary}
The rate function $\hat{\Lambda}^v$ is continuous.
\end{corollary}
\begin{proof}
Indeed, as a rate function, $\hat{\Lambda}^v$ is lower semi-continuous. Moreover, as $\Lambda^v$ is continuous, one can use similar arguments to~\cite[Corollary 4.6]{FZ17}, and deduce that $\hat{\Lambda}^v$ is upper semi-continuous, and hence is continuous.
\end{proof}

Before stating results on the small-time behaviour of options on integrated variance, we state that the integrated variance process $ RV(v)$ satisfies a large deviations principle on $\RR$ as t tends to zero, with speed $t^{-\beta}$ and rate function $\hat{\Lambda}^v(\E^\cdot)$.
Then, the small-time behaviour of such options can be obtained as an application of Corollary~\ref{integrated vol LDP}.

\begin{corollary}\label{cor:ImpliedVolrBergomi}
For log moneyness $k:=\log \frac{K}{RV(v)(0)} \neq 0$, the following equality holds true for Call options on integrated variance
\begin{equation}
\label{eq:OTM_IV}
\lim_{t \downarrow 0} t^\beta \log \mathbb{E} \left[ \left( RV(v)(t)-\E^k\right)^+\right] = - \mathrm{I}(k),
\end{equation}
where $\mathrm{I}$ is defined as as $\mathrm{I}(x) := \inf_{y>x} \hat{\Lambda}^v(\E^y)$ for $x >0$,
$\mathrm{I}(x) := \inf_{y<x} \hat{\Lambda}^v(\E^y)$ for $x <0$.

Similarly, for log moneyness $k:=\log \frac{K}{\sqrt{RV(v)(0)}} \neq 0$,
\begin{equation}
\label{eq:OTM_sqIV}
\lim_{t \downarrow 0} t^\beta \log \mathbb{E} \left[ \left( \sqrt{RV(v)(t)}-\E^k\right)^+\right] = -\bar{\mathrm{I}}(k),
\end{equation}
where $\bar{\mathrm{I}}$ is defined analogously as  $\bar{\mathrm{I}}(x) := \inf_{y>x} \hat{\Lambda}^v(\E^{2y})$ for $x >0$ and $\bar{\mathrm{I}}(x) := \inf_{y<x} \hat{\Lambda}^v(\E^{2y})$ for $x<0$.
\end{corollary}
\begin{proof}
The proof is postponed to Appendix \ref{3.7}.
\end{proof}

As with Call options on stock price processes, we can define and study the implied volatility of such options.
In the case of \eqref{integrated vol payoffs}(i) we define the implied volatility $\hat{\sigma}(T,k)$ to be the solution to
\begin{equation}\label{int vol implied vol}
\EE[(RV(v)(T)-e^k)^+ ]= C_{BS}( RV(v)(0), k, T, \hat{\sigma}(T,k) ),
\end{equation}
where $C_{BS}$ denotes the Call price in the Black-Scholes model.
Using Corollary \ref{cor:ImpliedVolrBergomi}, we deduce the small-time behaviour of the implied volatility $\hat{\sigma}$, as defined in \eqref{int vol implied vol}.

\begin{corollary}\label{cor:asymptoticImpliedVol}
The small-time asymptotic behaviour of the implied volatility is given by the following limit, for a log moneyness $k \neq 0$:
\begin{equation*}
\lim_{t \downarrow 0} t^{1-\beta} \hat{\sigma}^2(t,k)=:\hat{\sigma}^2(k) = \frac{k^2}{2\mathrm{I}(k)}.
\end{equation*}

\end{corollary}
\begin{proof}
The log integrated variance process $\log RV(v)$ satisfies a large deviations principle with speed $t^{-\beta}$ and rate function $\hat{\Lambda}^v(\E^\cdot)$, which is continuous.
Therefore, it follows that
$$ \lim_{t\downarrow 0} t^\beta \log \PP(RV(v)(t) \ge \E^k) = -\mathrm{I}(k).
$$
In the Black Scholes model, i.e. a geometric Brownian motion with $S_0=RV(v)(0)$ with constant volatility $\xi$, we have the following small-time implied volatility behaviour:
$$ \lim_{t \downarrow 0}  \xi^2 t \log\PP(RV(v)(t) \ge \E^k) = -\frac{k^2}{2}.
$$
We then apply Corollary \ref{cor:ImpliedVolrBergomi} and \cite[Corollary 7.1]{GL14}, identifying $\xi\equiv\hat{\sigma}(k,t)$, to conclude.
\end{proof}

\begin{remark}
Notice that the level of implied volatility in Corollary \ref{cor:asymptoticImpliedVol} has a power law behaviour as a function of time to maturity. This power law is of order $\sqrt{t^{\beta-1}}$, which is consistent with the at-the-money RV implied volatility results by Al\`os, Garc\'ia-Lorite and Muguruza \cite{AGM18}, where the order is found to be $t^{H-1/2}$ using Malliavin Calculus techniques. Recall that $\beta =2\alpha +1$, and $\alpha=H-1/2$ by Remark \ref{remark: H and alpha}.
\end{remark}

\subsection{Small-time results for the mixed rough Bergomi model}
Minor adjustments to Theorem \ref{LDP for vol process} give the following small-time result for the mixed variance process $v^{(\gamma,\nu)}$ introduced in Model~\eqref{eq:mixedvarianceProcess}; the proof is given in Appendix \ref{3.10}.

\begin{theorem}\label{truncated var ldp theorem}
The mixed variance process $( v_\eps^{(\gamma,\nu)})_{\eps>0}$ satisfies a large deviations principle on $\Cc(\Tt)$ with speed $\eps^{-\beta}$ and rate function
$$\Lambda^{(\gamma,\nu)}(\xx):= \inf\{\Lambda^Z(\frac{\eta}{\nu_1} \yy): \xx(\cdot) =v_0 \sum_{i=1}^n \gamma_i e^{\frac{\nu_i}{\nu_1}\yy(\cdot)} \},  $$
satisfying $\Lambda^{(\gamma,\nu)}(v_0)~=~0$.
\end{theorem}
By Remark \ref{remark: integrated variance remark}, we immediately get the following result for the small-time behaviour of the integrated mixed variance process $RV(v^{(\gamma,\nu)})$.
\begin{corollary}\label{integrated trunc var LDP}
The integrated mixed variance process $(RV\left(v^{(\gamma,\nu)}\right)(t))_{t\in \Tt}$ satisfies a large deviations principle on $\RR^*_+$ as $t$ tends to zero, with speed $t^{-\beta}$ and rate function
 $\tilde{\Lambda}^{(\gamma,\nu)}(\yy) := \inf \left\{ \Lambda^{(\gamma,\nu)}(\xx) : \yy = RV(\xx)(1) \right\}$, where $\tilde{\Lambda}^{(\gamma,\nu)}(v_0)=0$.
\end{corollary}

To get the small-time implied volatility result, analogous to Corollary \ref{cor:asymptoticImpliedVol}, we need the following Lemma, which is used in place of \eqref{qth moment of RV}. The remainder of the proof then follows identically.

\begin{lemma}\label{lemma: qth moment for mixed variance process}
For all $t\in\Tt$ and $q>1$ we have

$$\EE
\left[
\left(
RV\left(v^{(\gamma,\nu)}\right)(t)
\right)^q
\right] \leq \frac{v^q_0 n^q}{t^{q-1}} \exp \left(\frac{q^2(\nu^*)^2}{2\eta^2}\left(q^2-q\right)t^{2\alpha+1}\right),  $$
where $\nu^*=\max \{\nu_1,...,\nu_n\}$.
\end{lemma}

\begin{proof}
First we note that by H\"older's inequality $(\sum_{i=1}^n x_i)^q\leq n^{q-1} \sum_{i=1}^n (x_i)^q$, for $x_i>0$. Since, $\gamma_i\leq 1$ for $i=1,...,n$, we obtain
$$\left(RV\left(v^{(\gamma,\nu)}\right)(t)\right)^q
\le \frac{v_0^q}{t^q} n^{q-1} \sum_{i=1}^n\int_0^t
\mathbb{E}\left[ \exp\left(\frac{ q \nu_i}{\eta}Z_s-\frac{q\nu_i^2}{2 \eta^2 }s^{2\alpha+1}\right)
\right]
\D s\le  \frac{v_0^q}{t^{q-1}} n^{q-1} \sum_{i=1}^n
 \exp\left( \frac{\nu_i^2}{2 \eta^2 }\left(q^2-q \right)t^{2\alpha+1}\right).
$$
Choosing $\nu^*=\max \{\nu_1,...,\nu_n\}$ the result directly follows.

\end{proof}

\begin{corollary}\label{cor: small time option for mixed var process}
For log moneyness $k:=\log \frac{K}{RV(v^{(\gamma,\nu)})(0)} \neq 0$, the following equality holds true for Call options on integrated variance in the mixed rough Bergomi model:
\begin{equation}
\label{eq:OTM_IV mixed variance}
\lim_{t \downarrow 0} t^\beta \log \mathbb{E} \left[ \left( RV(v^{(\gamma,\nu)})(t)-\E^k\right)^+\right] = - \mathrm{I}(k),
\end{equation}
where $\mathrm{I}$ is defined as
$\mathrm{I}(x) := \inf_{y>x} \tilde{\Lambda}^{(\gamma,\nu)}(\E^y)$ for $x >0$,
$\mathrm{I}(x) := \inf_{y<x} \tilde{\Lambda}^{(\gamma,\nu)}(\E^y)$ for $x <0$.

Similarly, for log moneyness $k:=\log \frac{K}{\sqrt{RV(v^{(\gamma,\nu)})(0)}} \neq 0$,
\begin{equation}
\label{eq:OTM_sqIV mixed variance}
\lim_{t \downarrow 0} t^\beta \log \mathbb{E} \left[ \left( \sqrt{RV(v^{(\gamma,\nu)})(t)}-\E^k\right)^+\right] = -\bar{\mathrm{I}}(k),
\end{equation}
where $\bar{\mathrm{I}}$ is defined analogously as
$\bar{\mathrm{I}}(x) := \inf_{y>x} \tilde{\Lambda}^{(\gamma,\nu)}(\E^{2y})$ for $x >0$
and $\bar{\mathrm{I}}(x) := \inf_{y<x} \tilde{\Lambda}^{(\gamma,\nu)}(\E^{2y})$ for $x<0$.
\end{corollary}

\begin{proof}
Follows directly from Lemma \ref{lemma: qth moment for mixed variance process} and proof of Corollary \ref{cor:ImpliedVolrBergomi}.
\end{proof}

The small-time implied volatility behaviour for the mixed rough Bergomi model is then given by Corollary \ref{cor:asymptoticImpliedVol}, where the function $\mathrm{I}$ is defined in terms of the rate function $\tilde{\Lambda}^{(\gamma,\nu)}$, as in Corollary \ref{cor: small time option for mixed var process}, in this case.

\subsection{Small-time results for the multi-factor rough Bergomi model}
The small-time behaviour of the multi-factor rough Bergomi model~\eqref{eq:mixedvarianceMultiFactorProcess} can then be obtained; see Theorem \ref{Th:LDP_mixedOU} below; note that
$\Lambda^m$  is the rate function associated to the reproducing kernel Hilbert space of the measure induced by $(W_1,\cdots,W_m)$ on $\mathcal{C}(\Tt,\RR^m)$.
The proof is given in Appendix \ref{3.14}.

\begin{theorem}
\label{Th:LDP_mixedOU}
The variance process in the multi-factor rough Bergomi model $\left(v^{(\gamma,\nu,\Sigma)}_t\right)_{t\in \Tt}$ satisfies a large deviations principle on $\RR^*_+$ with speed $t^{-\beta}$ and rate function
$$\Lambda^{(\gamma,\nu,\Sigma)} (\yy)
 = \inf \left\{ \Lambda^m(\mathrm{x}) : \mathrm{x} \in \mathcal{H}_m, \mathrm{y} = v_0 \sum_{i=1}^n \gamma_i
 \exp \left( \frac{\nu^i}{\eta} \cdot \mathrm{L}_i \mathrm{x}(1)\right) \right\}, $$
satisfying $\Lambda^{(\gamma,\nu,\Sigma)} (v_0) = 0$.
\end{theorem}

As with the mixed variance process, Remark \ref{remark: integrated variance remark} gives us the following
small-time result for $RV(v^{(\gamma,\nu,\Sigma)})$ straight off the bat.
\begin{corollary}\label{integrated trunc var LDP}
The integrated variance process $(RV\left(v^{(\gamma,\nu,\Sigma)}\right)(t))_{t\in\Tt}$ in the multi-factor Bergomi model satisfies a large deviations principle on $\RR^*_+$ as $t$ tends to zero, with speed $t^{-\beta}$ and rate function
 $$\tilde{\Lambda}^{(\gamma,\nu,\Sigma)}(\yy) := \inf \left\{ \Lambda^{(\gamma,\nu,\Sigma)}(\xx) : \yy = RV(\xx)(1) \right\},$$ where $\tilde{\Lambda}^{(\gamma,\nu,\Sigma)}(v_0)=0$.
\end{corollary}

We now establish the small-time behaviour for Call options on realised variance in Corollary \ref{cor: small time option for multifactor rBerg process}, by adapting the proof of Corollary \ref{cor:ImpliedVolrBergomi} as in the previous subsection. To do so we use Lemma \ref{lemma: qth moment for multifactor rBergomi} in place of \eqref{qth moment of RV}.
Then we attain the small-time implied volatility behaviour for the multi-factor rough Bergomi model in Corollary \ref{cor:asymptoticImpliedVol}, where the function $\mathrm{I}$ is given by Corollary \ref{cor: small time option for multifactor rBerg process}.

\begin{lemma}\label{lemma: qth moment for multifactor rBergomi}
For all $t\in\Tt$ and $q>1$ we have

$$\EE
\left[
\left(
RV\left(v^{(\gamma,\nu,\Sigma)}\right)(t)
\right)^q
\right]\leq \frac{v^q_0 n^q}{t^{q-1}} \exp \left(\frac{(\nu^*)^2}{2 \eta^2 }\left(q^2-q \right)t^{2\alpha+1}\right),  $$
where $\nu^*=\max \{\nu_1,...,\nu_n\}$
\end{lemma}

\begin{proof}
First we note that by H\"older's inequality $(\sum_{i=1}^n x_i)^q\leq n^{q-1} \sum_{i=1}^n (x_i)^q$, for $x_i>0$. Since, $\gamma_i\leq 1$ for $i=1,...,n$, we obtain
$$\left(RV\left(v^{(\gamma,\nu,\Sigma)}\right)(t)\right)^q
\le \frac{v_0^q}{t^q} n^{q-1} \sum_{i=1}^n\int_0^t
\mathbb{E}\left[ \mathcal{E}\left(\frac{\nu^i}{\eta}\cdot\mathrm{L}_i \mathcal{Z}_s\right)^q
\right]
\D s\le  \frac{v_0^q}{t^{q-1}} n^{q-1} \sum_{i=1}^n
 \exp\left( \frac{\nu_i^2}{2 \eta^2 }\left(q^2-q \right)t^{2\alpha+1}\right).
$$
Choosing $\nu^*=\max \{\nu_1,...,\nu_n\}$ the result directly follows.
\end{proof}

\begin{corollary}\label{cor: small time option for multifactor rBerg process}
For log moneyness $k:=\log \frac{K}{RV(v^{(\gamma,\nu,\Sigma)})(0)} \neq 0$, the following equality holds true for Call options on integrated variance in the multi-factor rough Bergomi model:
\begin{equation}
\label{eq:OTM_IV multifactor }
\lim_{t \downarrow 0} t^\beta \log \mathbb{E} \left[ \left( RV(v^{(\gamma,\nu,\Sigma)})(t)-\E^k\right)^+\right] = - \mathrm{I}(k),
\end{equation}
where $\mathrm{I}$ is defined as $\mathrm{I}(x) := \inf_{y>x} \tilde{\Lambda}^{(\gamma,\nu,\Sigma)}(\E^y)$ for $x >0$,
$\mathrm{I}(x) := \inf_{y<x} \tilde{\Lambda}^{(\gamma,\nu,\Sigma)}(\E^y)$ for $x <0$.

Similarly, for log moneyness $k:=\log \frac{K}{\sqrt{RV(v^{(\gamma,\nu,\Sigma)})(0)}} \neq 0$,
\begin{equation}
\label{eq:OTM_sqIV multifactor}
\lim_{t \downarrow 0} t^\beta \log \mathbb{E} \left[ \left( \sqrt{RV(v^{(\gamma,\nu,\Sigma)})(t)}-\E^k\right)^+\right] = -\bar{\mathrm{I}}(k),
\end{equation}
where $\bar{\mathrm{I}}$ is defined analogously as
$\bar{\mathrm{I}}(x) := \inf_{y>x} \tilde{\Lambda}^{(\gamma,\nu,\Sigma)}(\E^{2y})$ for $x >0$
and $\bar{\mathrm{I}}(x) := \inf_{y<x} \tilde{\Lambda}^{(\gamma,\nu,\Sigma)}(\E^{2y})$ for $x<0$.
\end{corollary}

\begin{proof}
Follows directly from Lemma \ref{lemma: qth moment for multifactor rBergomi} and the proof of Corollary \ref{cor:ImpliedVolrBergomi}.
\end{proof}

\section{Numerical results}\label{sec: numerics}
In this section we present numerical results for each of the three models given in Section \ref{sec:models}. We also analyse the effect of each parameters in the implied volatility smile.
Numerical experiments and codes are provided on \href{https://github.com/amuguruza/LDP-VolOptions}{GitHub: LDP-VolOptions }.
\subsection{RV smiles for rough Bergomi}
We begin with numerical results for the rough Bergomi model \eqref{rough Bergomi} using Corollary \ref{cor:asymptoticImpliedVol}. For the detailed numerical method we refer the reader to Appendix \ref{appendix: numerics}.
In Figure \ref{fig:rateFunction}, we represent the rate function given in Corollary \ref{integrated vol LDP}, which
 is the fundamental object to compute numerically. In particular, we notice that $\hat{\Lambda}^v$ is convex; a rigorous mathematical proof of this statement is left for future research.

\begin{figure}[h]
\hspace*{-30mm}
\includegraphics[scale=0.43]{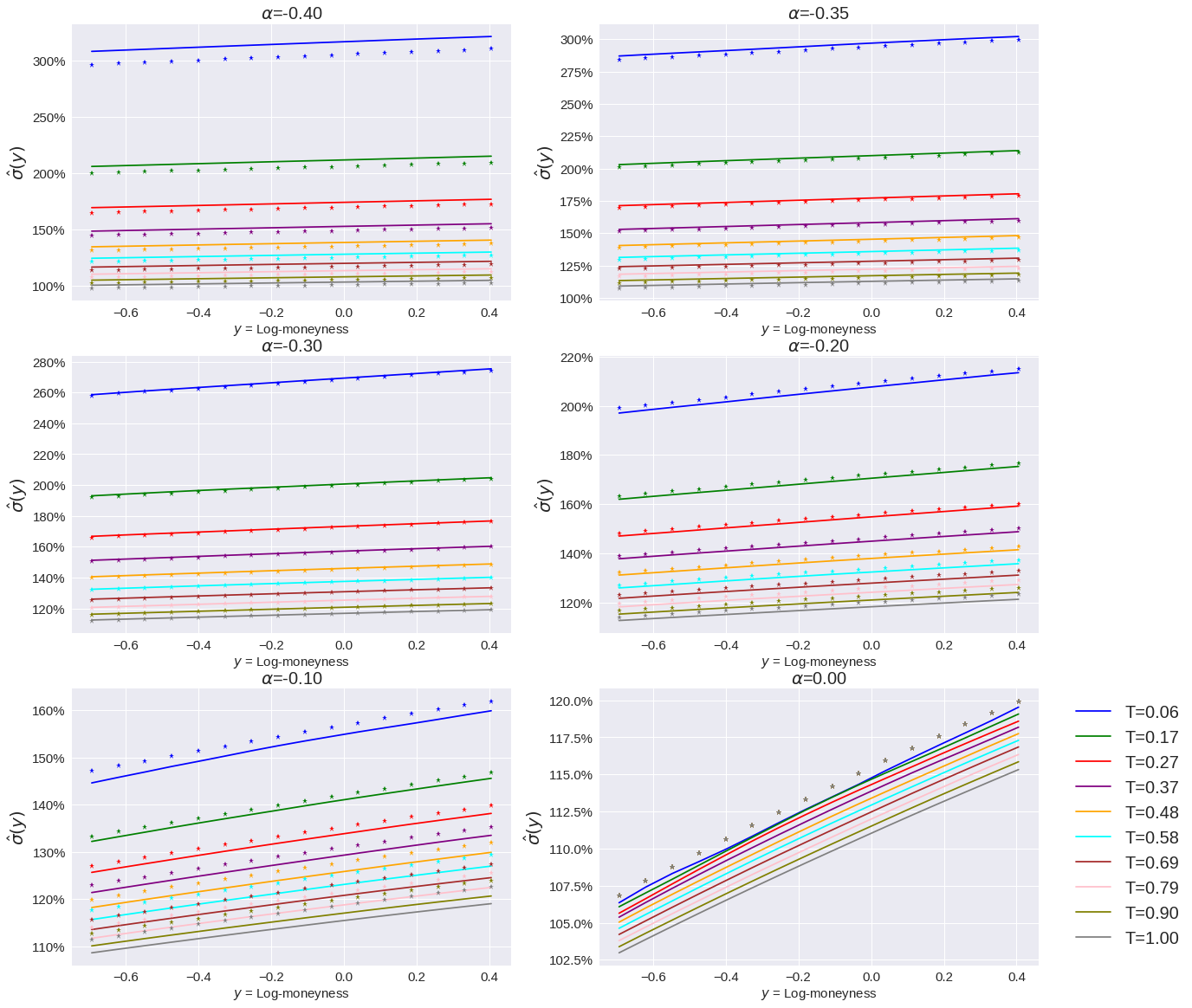}
\caption{Comparison of Monte Carlo computed  implied volatilities (straight lines) and LDP based implied volatilities (stars), in the rough Bergomi model, for different values of $\alpha$ and maturities $T$. We set $\eta=2$ and $v_0=0.04$; for Monte Carlo we use $200,000$ simulations and $\Delta t =\frac{1}{1008}$.}
\label{fig:MC vs. LDP}
\end{figure}

\begin{figure}[h]
\centering
\includegraphics[scale=0.38]{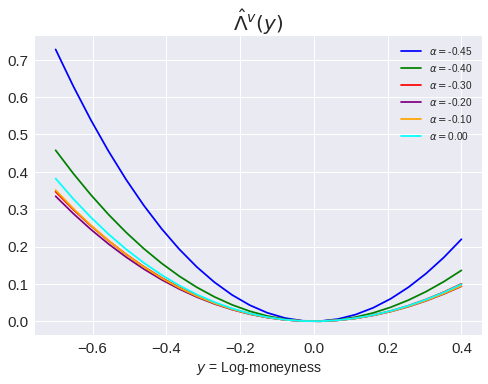}
\caption{Rate function $\hat{\Lambda}^v$ for different values of $\alpha$.}
\label{fig:rateFunction}
\end{figure}

More interestingly, in Figure \ref{fig:MC vs. LDP} we provide a comparison of Corollary \ref{cor:asymptoticImpliedVol} with respect to a benchmark generated by Monte Carlo simulations, and see all smiles to follow a linear trend.
In particular, we notice that Corollary \ref{cor:asymptoticImpliedVol} provides a surprisingly accurate estimate, even for relatively large maturities.
As a further numerical check, in Figure \ref{fig:AGM vs. LDP} we compare our results with the close-form at-the-money asymptotics given by Al\`os, Garc\'ia-Lorite and Muguruza \cite{AGM18} and once again find the correct convergence, suggesting a consistent numerical framework.

\begin{figure}[h!]
\hspace*{-1.0cm}
\includegraphics[scale=0.43]{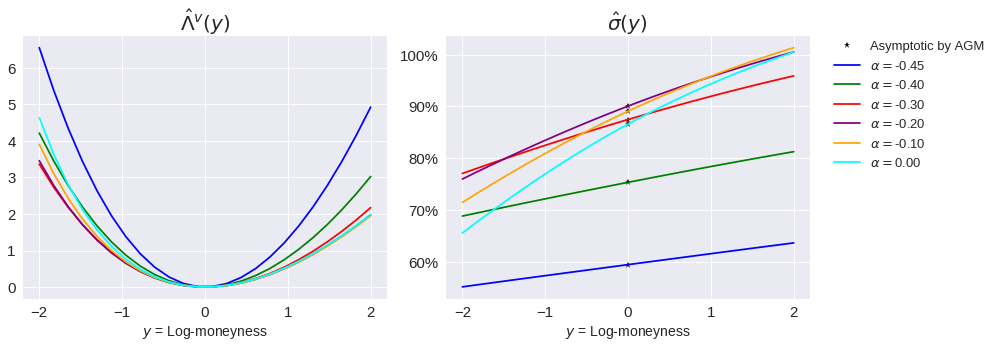}
\caption{Comparison of Al\`os, Garc\'ia-Lorite and Muguruza \cite{AGM18} at-the-money  implied volatility asymptotics  and LDP based implied volatilities for different values of $\alpha$, in the rough Bergomi model, with $\eta=1.5$  and $v_0=0.04$.}
\label{fig:AGM vs. LDP}
\end{figure}

\subsection{RV smiles for mixed rough Bergomi}\label{section: numerics mixed case}
We now consider the mixed rough Bergomi model \eqref{eq:mixedvarianceProcess} in a simplified form given by
$v_t=v_0 \left(\gamma_1 \mathcal{E}(\nu_1 Z_t)+\gamma_2 \mathcal{E}(\nu_2 Z_t)\right)$.
In Figure \ref{fig:roughBergomiSmilesFixedVolOfVol}, we observe that a constraint of the type $\gamma_1\nu_1+\gamma_2\nu_2=2$ in  the mixed variance process \eqref{eq:mixedvarianceProcess} allows us to fix the at-the-money implied volatility at a given level, whilst generating different slopes for different values of $(\nu_1,\nu_2,\gamma_1,\gamma_2)$; as in Figure \ref{fig:MC vs. LDP}, we see that the smiles generated follow a linear trend.
This is again consistent with the results found in \cite{AGM18}.

\begin{figure}[h!]
\centering
\includegraphics[scale=0.43]{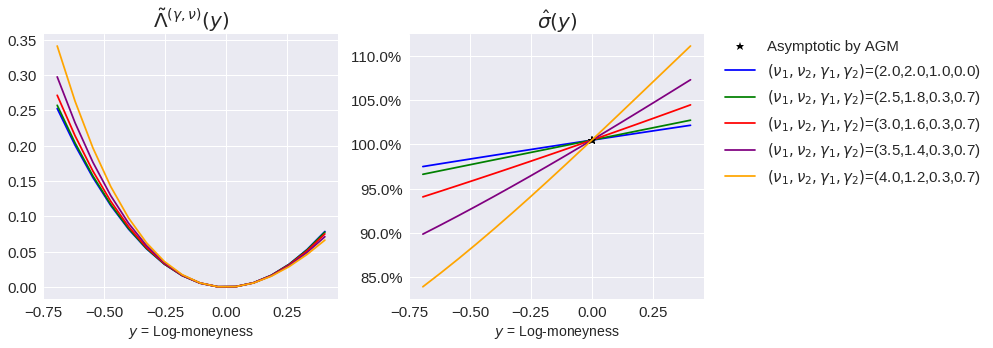}
\caption{Comparison of LDP based implied volatilities for different values of $(\nu_1,\nu_2,\gamma_1,\gamma_2)$ in the mixed rough Bergomi process \eqref{eq:mixedvarianceProcess} such that $\gamma_1\nu_1~+~\gamma_2\nu_2~=~2$, with $\alpha=-0.4$.}
\label{fig:roughBergomiSmilesFixedVolOfVol}
\end{figure}
\clearpage

\begin{remark}
At this point it is important to note that the mixed rough Bergomi model \ref{eq:mixedvarianceProcess} allows both the at-the-money implied volatility and its skew to be controlled through $(\gamma,\nu)$, whilst in the rough Bergomi model \eqref{rough Bergomi} there is not enough freedom to arbitrarily fit both quantities.
\end{remark}
\subsection{Linearity of smiles and approximation of the RV density }
Remarkably, we observe a linear pattern in Figures \ref{fig:MC vs. LDP}-\ref{fig:roughBergomiSmilesFixedVolOfVol} for around the money strikes in the (mixed) rough Bergomi model. By assuming this linear implied volatility smile in log-moneyness space, in Appendix \ref{app:linear approximation} we are able  to provide an approximation scheme for the realised variance density $\psi_{RV}(x,T),\; x,T\geq 0$ (see Proposition \ref{prop:density}) . This in turn, allows us to compute the price of a volatility swap by using,
$$\mathcal{P}_{VolSwap(T)}=\int_0^{\infty} \sqrt{x}\psi_{RV}(x,T) \D x,$$
where $\mathcal{P}_{VolSwap(T)}$ is the price of a Volatility Swap with maturity $T$.  Figure \ref{fig:VolSwap} presents the results of such approximation technique that turns out to be surprisingly accurate. We believe that this approximation could be useful to practitioners and more details are provided in Appendix \ref{app:linear approximation}.

\subsection{RV smiles for mixed multi-factor rough Bergomi}\label{section: numerics correlated case}
We conclude our analysis by introducing the correlation effect in the implied volatility smiles, by considering the mixed multi-factor rough Bergomi model \eqref{eq:mixedvarianceMultiFactorProcess}. We shall consider the following two simplified models for instantaneous variance
\begin{align}\label{eq:addingfactors}
v_t&=\mathcal{E}\left(\nu \int_0^t (t-s)^\alpha\D W_s+ \eta\left(\rho \int_0^t (t-s)^\alpha\D W_s+\sqrt{1-\rho^2}\int_0^t (t-s)^\alpha\D W^\perp_s\right) \right),\\\label{eq:mixingfactors}
v_t&=\frac{1}{2}\left(\mathcal{E}\left(\nu \int_0^t (t-s)^\alpha\D W_s\right)+\mathcal{E}\left( \eta\left(\rho \int_0^t (t-s)^\alpha\D W_s+\sqrt{1-\rho^2}\int_0^t (t-s)^\alpha\D W^\perp_s\right) \right)\right),
\end{align}
where $W$ and $W^\perp$ are independent standard Brownian motions and $\nu,\eta>0$.

On one hand, Figure \ref{fig:CorrelatedSmiles1} shows the implied volatility smiles corresponding to \eqref{eq:addingfactors}.
We conclude that adding up correlated factors inside the exponential does not change the behaviour of implied volatility smiles, and they still have a linear form around the money.
Moreover, in this context \cite{AGM18} results still hold and we provide the asymptotic benchmark in Figure \ref{fig:CorrelatedSmiles1} to support our numerical scheme.
On the other hand, Figure \ref{fig:CorrelatedSmiles2} shows the implied volatility smiles corresponding to \eqref{eq:mixingfactors}, which are evidently non-linear around the money in the negatively correlated cases.
Consequently, we can see that having a sum of exponentials, each one driven by a different (fractional) Brownian motion does indeed affect the behaviour of the convexity in the implied volatility around the money. We further superimpose a linear trend on top of the smiles in Figure \ref{fig:CorrelatedSmilesLinearApprox2} to clearly show the effect of correlation in the convexity of the smiles.

\begin{figure}[h!]
 \hspace*{-15mm}
\includegraphics[scale=0.55]{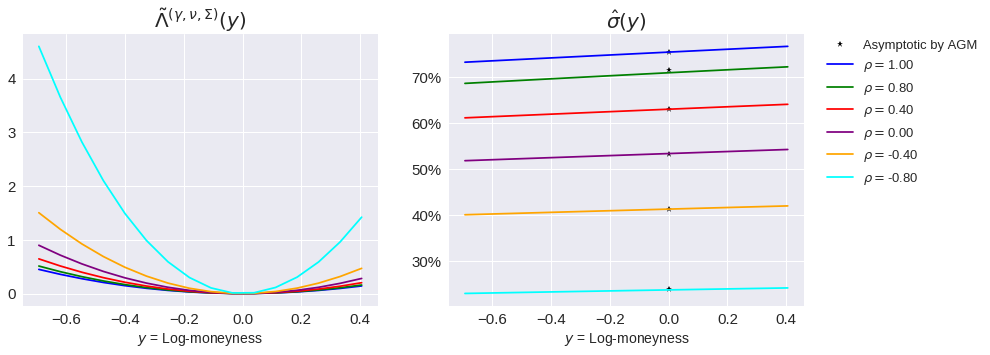}
\caption{Rate function and corresponding implied volatilities for the model \eqref{eq:addingfactors},
with $(\alpha,\nu,\eta)=(-0.4,0.75,0.75)$. }
\label{fig:CorrelatedSmiles1}
\end{figure}

\clearpage
\begin{remark}
Note here that models such as the 2 Factor Bergomi (and its mixed version) \cite{Bergomi3,Ber16} have the same small-time limiting behaviour as \eqref{eq:addingfactors} due to Proposition \ref{pp:exp_equi L},
where we set $L=L_{\text{Gamma}}$ and $\alpha=0$, meaning that their smiles follow a linear trend.
Despite the empirical evidence given in Figure \ref{VIX plots}, the commitment to linear smiles from a modelling perspective is strong.
The simplified mixed multi-factor rough Bergomi model \eqref{eq:mixingfactors}, however, is sufficiently flexible to generate both linear ($\rho \ge 0$) and nonlinear ($\rho < 0$) smiles.
\end{remark}

\begin{figure}[h!]
\hspace*{-1.0cm}
\includegraphics[scale=0.5]{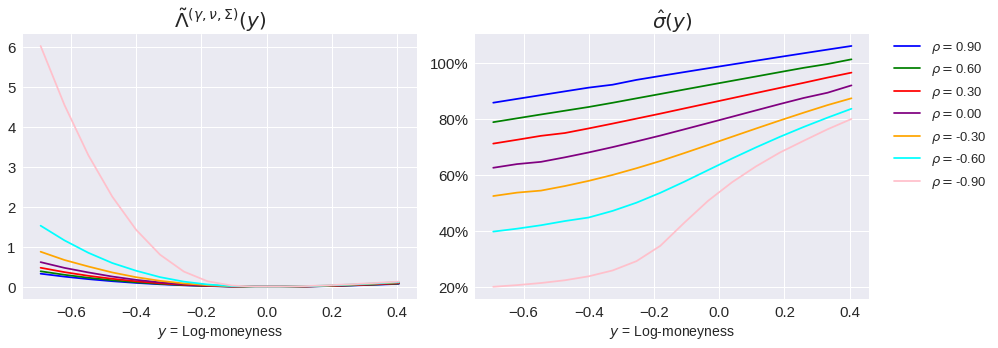}
\caption{Rate function and corresponding implied volatilities for the model \eqref{eq:mixingfactors} with
$(\alpha,\nu,\eta) =(-0.4,1.0,3.0)$.}
\label{fig:CorrelatedSmiles2}
\end{figure}
\begin{figure}[h!]
\hspace*{-1.0cm}
\includegraphics[scale=0.55]{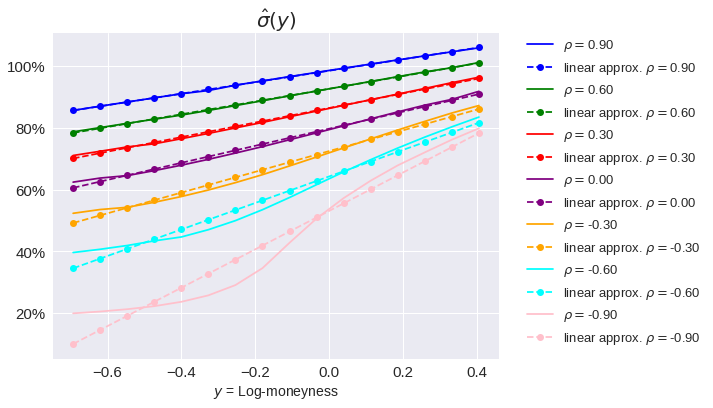}
\caption{Implied volatilities and superimposed linear smiles for the model \eqref{eq:mixingfactors},
with $(\alpha,\nu,\eta) =(-0.4,1.0,3.0)$.}
\label{fig:CorrelatedSmilesLinearApprox2}
\end{figure}

\section{Options on VIX}\label{section: VIX options}
Although options on realised variance are the most natural core modelling object for stochastic volatility models, in practice the most popular variance derivative is the VIX. In this section we therefore turn our attention to the VIX and VIX options and study their asymptotic behaviour.
For this section, we fix $\Tt := [0,T]$.
Let us now consider the following general model $(v_t)_{t\ge 0}$ for instantaneous variance:
\begin{equation}\label{eq:varianceGeneralKernel}
v_t=\xi_0(t)\mathcal{E}\left( \int_0^t g(t,s) \D W_s\right).
\end{equation}
Then, the VIX process is given by
$$\text{VIX}_T=\sqrt{\frac{1}{\Delta}\int_T^{T+\Delta}\mathbb{E}[v_t|\mathcal{F}_T] \D t}.$$
We introduce the following stochastic process $(\VgT)_{t \in[T,T+\Delta]}$, for notational convenience,  as
\begin{equation}\label{vix type definition}
V^{g,T}_t:= \int_0^T g(t,s)\D W_s,
\end{equation}
and assume that the mapping $s \mapsto g(t,s) $ is in $L^2[0,T]$ for all $t \in [T,T+\Delta]$ such that the stochastic integral in \eqref{vix type definition} is well-defined.
\begin{proposition}\label{prop:rough BergomiVIX}
The VIX  dynamics in model \eqref{eq:varianceGeneralKernel}, where the volatility of volatility is denoted by $\nu$, are given by
$$
\text{VIX}_{T, \nu}^2
  := \frac{1}{\Delta}\int_T^{T+\Delta}\xi_0(t)\exp\left(\nu V_t^{g,T}-\frac{\nu^2}{2}\mathbb{E}[(V_t^{g,T})^2] \right)\D t.
$$

\end{proposition}
\begin{proof}
Follows directly from \cite[Proposition 3.1]{JMM18}.
\end{proof}

We now define the following $L^2[0,T]$ operator $\IgT:L^2[0,T] \rightarrow \Cc[T,T+\Delta]$, and space $\HgT$ as
\begin{equation}
\IgT f(\cdot):= \int_0^T g(\cdot,s) f(s) \D s, \quad
\HgT:= \{\IgT f: f\in L^2[0,T] \},
\end{equation}
where the space $\HgT$ is equipped with the following inner product
$ \langle \IgT f_1, \IgT f_2 \rangle_{\HgT}:= \langle f_1, f_2, \rangle_{L^2[0,T]}. $
Note that the function $g$ must be such that the operator $\IgT $ is injective so that that inner product
$\langle \cdot, \cdot \rangle_{\HgT}$ on
$\HgT$ is well-defined.

\begin{proposition}\label{rkhs of vix type process}
Assume that there exists $ h \in L^2[0,T]$ such that $\int_0^{\eps} |h(s)|\D s< +\infty$ for some $\eps>0$  and
$g(t,\cdot)=h(t-\cdot)$ for any $t \in [T,T+\Delta]$.
Then, the space $\HgT$ is the reproducing kernel Hilbert space for the process $( V^{g,T}_t)_{t\in [T, T+\Delta]}$.
\end{proposition}

\begin{proof}
See Appendix \ref{5.2}.
\end{proof}

\begin{theorem}\label{VIX ldp on C}
For any $\gamma>0$, the sequence of stochastic processes $(\eps^{\gamma/2} \VgT )_{\eps >0}$ satisfies a large deviations principle on
$\CCT $ with speed $\eps^{-\gamma}$ and rate function $\Lambda^V$, defined as
\begin{equation}\label{vix rate function}
\Lambda^V(\xx) :=
\begin{cases}
\displaystyle \frac{1}{2}\|\xx\|_{\HgT}^2, & \text{if }\xx\in \HgT,\\
+ \infty, & \text{otherwise.}
\end{cases}
\end{equation}
\end{theorem}

\begin{proof}
Direct application of the generalised Schilder's Theorem \cite[Theorem 3.4.12]{DS89}.
\end{proof}

\begin{remark}\label{VgT large strike behaviour}
We now introduce a Borel subset of $\CCT$, defined as $A := \{ g\in\CCT :g(x)~\ge~1 \text{ for all } x\in\RR \}$.
Then, by a simple application of Theorem~\ref{VIX ldp on C} and using that the rate function $\Lambda^V$ is continuous on $A$, we can obtain then obtain the following tail behaviour of the process $V^{g,T}$:
\begin{equation}\label{VIX large strike behaviour}
\lim_{\eps \downarrow 0} \eps^\gamma \log \mathbb{P} \left( V^{g,T}_t \ge \frac{1}{\eps^{\gamma/2}}\right)
= - \inf_{g\in A} \Lambda^V(g),
\end{equation}
for any $\gamma>0$ and $t\in[T,T+\Delta]$.
\end{remark}

\begin{remark}
Let us again fix the kernel $g$ as the rough Bergomi kernel and denote the corresponding reproducing kernel Hilbert space by $\mathscr{H}^{\eta, \alpha, T}$ and the corresponding process $\VgT$ as $V^{\eta,\alpha,T}$.
If $\xx \in \mathscr{H}^{\eta, \alpha, T}$ it follows that there exists $f\in L^2[0,T]$ such that
$\xx(t)~=~\int_0^T \eta \sqrt{2\alpha +1}(t-s)^\alpha f(s)\D s $ for all $t \in [T,T+\Delta]$.
Clearly, it follows that $\xx \in \mathscr{H}^{a \eta, \alpha, T} $ for any $a>0$, as $f\in L^2[0,T]$ implies that $\frac{1}{a} f=:f_a \in L^2[0,T]$.
We can compute the norm of $\xx$ in each of these spaces to arrive at the following isometry:
\begin{equation}\label{norm isometry for rBerg vix}
\Vert \xx \Vert^2_{\mathscr{H}^{a \eta, \alpha, T}}  = \Vert f_a \Vert^2_{L^2[0,T]} =
\frac{1}{a^2} \int_0^T f^2(s)\D s = \frac{1}{a^2} \Vert \xx \Vert^2_{\mathscr{H}^{ \eta, \alpha, T}}.
\end{equation}
We may now amalgamate \eqref{vix rate function}, \eqref{VIX large strike behaviour}, and \eqref{norm isometry for rBerg vix} to arrive at the following statement, which tells us how the large strike behaviour scales with the vol-of-vol parameter $\eta$ in the rough Bergomi model:
\begin{equation}\label{vol-of-vol effect on large strike}
\lim_{\eps \downarrow 0} \eps^\gamma \log P \left(V^{a\eta, \alpha, T}_t \ge {\frac{1}{\eps^{\gamma/2}}} \right)
=
\lim_{\eps \downarrow 0} \eps^\gamma \log \left( P \left(V^{\eta, \alpha, T}_t \ge {\frac{1}{\eps^{\gamma/2}}} \right)^{1/a^2}\right)
\end{equation}
Indeed, \eqref{vol-of-vol effect on large strike} tells us precisely how increasing the vol-of-vol parameter $\eta$ multiplicatively by a factor $a$ in the rough Bergomi model increases the probability that the associated process $\VgT$ will exceed a certain level.
\end{remark}

Before stating the main theorem of this section, we first define the following rescaled process:
\begin{equation}\label{eq: rescaled VgT processes}
V^{g,T,\eps}_t:= \eps^{\gamma /2}\VgT_t, \quad
\widetilde{V}^{g,T,\eps }_t:= V^{g,T,\eps}_t - \frac{\eps^\gamma}{2}\int_0^tg^2(t,u)\D u + \eps^{\gamma /2},
\end{equation}
for $\eps \in [0,1]$, $t \in [T,T+\Delta]$.
We also define the following $\Cc \left([T,T+\Delta] \times [0,1] \right) $ operators $\varphi_{1,\xi_0}, \varphi_2$, which map to $\Cc \left([T,T+\Delta] \times [0,1] \right) $ and $\Cc[0,1]$ respectively, as
\begin{equation}\label{varphi operators}
(\varphi_{1,\xi_0}f)(s,\eps):= \xi_0(s)\exp(f(s,\eps)), \quad
(\varphi_2g)(\eps):=\frac{1}{\Delta} \int_T^{T+\Delta} g(s,\eps)\D s.
\end{equation}
Note that in the definition of $\varphi_{1,\xi_0}$ in \eqref{varphi operators} we assume $\xi_0$ to be a continuous, single valued, and strictly positive function on $[T,T+\Delta]$.
This then implies that for every $s\in [T,T+\Delta]$, the map $\eps \mapsto (\varphi_{1,\xi_0} f)(s,\eps)$ is a bijection and hence has an inverse, denoted by $\varphi^{-1}_{1,\xi_0}$, which is defined as
$ (\varphi^{-1}_{1,\xi_0}f)(s,\eps):=\log \left( \frac{f(s,\eps)}{\xi_0(s)} \right)
$.

\begin{theorem}\label{them: rescaled VIX ldp}
For any $\gamma>0$,
the sequence of rescaled VIX processes
$(e^{\eps^{\gamma/2}}\text{VIX}_{T,\eps^{\gamma/2}})_{\eps \in [0,1]}
$
satisfies a pathwise large deviations principle on $\Cc[0,1]$ with speed $\eps^{-\gamma}$ and rate function
$$\Lambda^{\text{VIX}}(\xx):= \inf_{s\in[T,T+\Delta]}
\left\{
\Lambda^V\left(\log\left(\frac{\yy(s,\cdot)}{\xi_0(s)} \right)  \right)
:
\xx(\cdot) = (\varphi_2\yy)(\cdot)
\right\}.
$$
\end{theorem}

\begin{proof}
See Appendix \ref{5.6}.
\end{proof}

\begin{remark}
Using Theorem \ref{them: rescaled VIX ldp}, we can deduce the small-noise, large strike behaviour of VIX options.
Indeed, for the Borel subset $A$ of $\CCT]$ introduced in Remark \ref{VgT large strike behaviour} we have that
$$ \lim_{\eps \downarrow 0} \eps^\gamma \log
\PP\left(VIX_{T, \eps^{\gamma/2}} \ge e^{-\eps^{\gamma/2}} \right) = - \inf_{g \in A} \Lambda^{\text{VIX}}(g),
$$
for any $\gamma>0$.
\end{remark}
\section{conclusions}
In this paper we have characterised, for the first time, the small-time behaviour of options on integrated variance in rough volatility models, using large deviations theory.
Our approach has a solid theoretical basis, with very convincing corresponding numerics, which agree with observed market phenomenon and the theoretical results attained by Al\`os, Garc\'ia-Lorite and Muguruza \cite{AGM18}. Both the theoretical and the numerical results hold for each of the three rough volatility models presented, whose complexity increases. Any of the three, with our corresponding results, would be suitable for practical use; the user would simply chose the level of complexity needed to satisfy their individual needs.
Note also that the theoretical results are widely applicable, and one could very easily adapt results presented in this paper to other models where the volatility process also satisfies a large deviations principle, and whose rate function can be computed easily and accurately.

Perhaps surprisingly, we have discovered that lognormal models such as rough Bergomi \cite{BFG16}, 2 Factor Bergomi \cite{Bergomi3,Ber16} and mixed versions thereof, generate linear smiles around the money for options on realised variance in log-space. This is, at the very least, a property to be taken into account when modelling volatility derivatives and, to our knowledge, has never been addressed or commented on in previous works. Whether such an assumption is realistic or not, we have in addition provided an explicit way to construct a model that generates non-linear smiles should this be required or desired.

We have also proved a pathwise large deviations principle for rescaled VIX processes, in a fairly general setting with minimal assumptions on the kernel of the stochastic integral used to define the instantaneous variance; these results are then used to establish the small-noise, large strike asymptotic behaviour of the VIX.
The current set up does not allow us to deduce the small-time VIX behaviour from the pathwise large deviations principle, but this would be a very interesting area for future research. Our numerical scheme would most likely give a good approximation for the rate function and corresponding small-time VIX smiles.

\appendix

\section{Approximating the density of realised variance in the mixed rough Bergomi model}\label{app:linear approximation}
In light of the numerical results shown in Section \ref{sec: numerics} (see Figures \ref{fig:MC vs. LDP}-\ref{fig:CorrelatedSmiles1}) we identify a clear linear trend in the implied volatility smiles generated by both the rough Bergomi and mixed rough Bergomi models. Therefore, it is natural to postulate the following conjecture/approximation of log-linear smiles.

\begin{assumption}\label{linearImpliedVol}
The implied volatility of realised variance options in the mixed rough Bergomi \eqref{eq:mixedvarianceProcess} model is linear in log-moneyness, and takes the following form:

$$\hat{\sigma}(K,T)=\left(T^\beta \left(a(\alpha,\gamma,\nu)+b(\alpha,\gamma,\nu)\log\left(\frac{K}{RV(v)(0)}\right)\right)\right)^+$$
where
\begin{align*}
a(\alpha,\gamma,\nu)&=\frac{\sqrt{2\alpha +1}\sum_{i=1}^n \gamma_i \nu_i}{(\alpha +1)\sqrt{2\alpha + 3}},\\
b(\alpha,\gamma,\nu)&= \sqrt{2\alpha+1}\left(\frac{\sum_{i=1}^n \gamma_i \nu_i^2}{\sum_{i=1}^n \gamma_i \nu_i}\mathcal{I}(\alpha)(2\alpha+3)^{3/2}(\alpha+1)-\frac{\sum_{i=1}^n \gamma_i \nu_i}{(2\alpha+2)\sqrt{(2\alpha+3)}}\right),
\end{align*}
with
$$\mathcal{I}(\alpha)=\frac{\bigg(\displaystyle\sum_{n=0}^{\infty}\frac{(\alpha)_n}{(\alpha+2)_n}\frac{1-2^{-2\alpha-3-n}}{2\alpha+3+n}+\sum_{n=0}^{\infty}(-1)^n
\frac{(-\alpha)_n (\alpha+1)}{(\alpha+2+n)n!}
\frac{\hat{F}(n,1)-2^{n-1/2-n}\hat{F}(n,1/2)}{\alpha+1-n}\bigg)}{(\alpha+1)(4\alpha+5)}$$
such that $\hat{F}(n,x)=\text{ }_2F_1(-n-2\alpha-2,\alpha+1-n,\alpha+2-n,x)$ and $(x)_n=\displaystyle\prod_{i=0}^{n-1}(x+i)$ represents the rising Pochhammer factorial.
\end{assumption}

\begin{remark}
The values of the constants $a(\alpha,\gamma,\nu)$ and $b(\alpha,\gamma,\nu)$ in Assumption \ref{linearImpliedVol}, which give the level and slope of the implied volatility respectively,  are given in \cite[Example 24 and Example 27]{AGM18} respectively; we generalise to $n$ factors.
These results are given in terms of the Hurst parameter $H$; to avoid any confusion we will continue with our use of $\alpha$.
Recall that, by Remark \ref{remark: H and alpha}, $\alpha=H-1/2$.
\end{remark}

\begin{proposition}\label{prop:density}
Under Assumption \ref{linearImpliedVol} , the density of $RV(v^{(\gamma,\nu)})(T)$ is given by

$$\psi_{RV}(x,T)=-\phi(d_2(x))\frac{\partial d_1(x)}{\partial x}\left(a(\alpha,\gamma,\nu) T^{\alpha+1/2}d_1(x)+1\right),\quad x\geq 0$$
where $d_1(x)=\frac{\log(v_0)-\log(x)}{\hat{\sigma}(x,T)\sqrt{T}}+\frac{1}{2}\hat{\sigma}(x,T)\sqrt{T}$, $d_2(x)=d_1(x)- \hat{\sigma}(x,T)\sqrt{T}$ for $x\ge0$ and $\phi(\cdot)$ is the standard Gaussian probability density function.
\end{proposition}
\begin{proof}
Let us denote
$$C(K,T):=\mathbb{E}[(RV(v^{(\gamma,\nu)})(T)-K)^+].$$
The well-known Breeden-Litzenberger formula \cite{BL78} tells us that
$$\frac{\partial^2 C(x,T)}{\partial x^2}\bigg|_{x=K}=\psi_{RV}(K,T).$$
Under Assumption \ref{linearImpliedVol}, we have that
$$C(K,T)=BS(v_0,\hat{\sigma}(K,T),K,T)$$
where $BS(v_0,\sigma,K,T)=v_0\Phi(d_1)-K\Phi(d_2)$ is the Black-Scholes Call pricing formula with $\Phi$ the standard Gaussian cumulative distribution function. Then, differentiating $C$ with respect to the strike gives

$$\frac{\partial C(x,T)}{\partial x}\bigg|_{x=K}=v_0\phi(d_1(K)) \frac{\partial d_1(x)}{\partial x}\bigg|_{x=K}-x\phi(d_2(K)) \frac{\partial d_2(x)}{\partial x}\bigg|_{x=K}-\Phi(d_2(K))$$
where
\begin{eqnarray*}\frac{\partial d_1(x)}{\partial x}&&=\frac{-\hat{\sigma}(x,T)+\log(x/v_0)a(\alpha,\gamma,\nu) T^{\alpha}}{x\hat{\sigma}(x,T)^2\sqrt{T}}+\frac{1}{2}\frac{a(\alpha,\gamma,\nu) T^{\alpha+1/2}}{x}\\&&=\frac{-b(\alpha,\gamma,\nu) T^{\alpha}}{x\hat{\sigma}(x,T)^2\sqrt{T}}+\frac{1}{2}\frac{a(\alpha,\gamma,\nu) T^{\alpha+1/2}}{x}\end{eqnarray*}
and $$\frac{\partial d_2(x)}{\partial x}=\frac{\partial d_1(x)}{\partial x}-\frac{a(\alpha,\gamma,\nu) T^{\alpha+1/2}}{x}.$$
Using the well known identity $v_0\phi(d_1(x))=x \phi(d_2(x))$, proved in Appendix~\ref{appendix: Proof_A1},
we further simplify
$$\frac{\partial C(K,T)}{\partial K}=v_0\phi(d_1(K))\left(\frac{a(\alpha,\gamma,\nu) T^{\alpha+1/2}}{K}\right)-\Phi(d_2(K)).$$
Differentiating again we obtain,

$$\psi_{RV}(K,T)=-v_0\phi(d_1(K))\frac{a(\alpha,\gamma,\nu) T^{\alpha+1/2}}{K}\left(d_1(K)\frac{\partial d_1(x)}{\partial x}\bigg|_{x=K}+\frac{1}{K}\right)-\phi(d_2(K))\frac{\partial d_2(x)}{\partial x}\bigg|_{x=K}.$$
Then, by using $v_0\phi(d_1(x))=x \phi(d_2(x))$, we find that
$$\psi_{RV}(K,T)=-\phi(d_2(K))\left(a(\alpha,\gamma,\nu) T^{\alpha+1/2}\left(d_1(x)\frac{\partial d_1(x)}{\partial x}\bigg|_{x=K}+\frac{1}{K}\right)+\frac{\partial d_2(x)}{\partial x}\bigg|_{x=K}\right),$$
which we further simplify to
$$\psi_{RV}(K,T)=-\phi(d_2(K))\frac{\partial d_1(x)}{\partial x}\bigg|_{x=K}\left(a(\alpha,\gamma,\nu) T^{\alpha+1/2}d_1(K)+1\right),$$
and the result then follows.
Note that the density $\psi_{RV}(\cdot,T)$ is indeed continuous for all $T>0$.
\end{proof}
\begin{remark}
Note that Proposition \ref{prop:density} gives the density of $RV(v^{(\gamma,\nu)})(T)$ in closed-form. In addition, Proposition \ref{prop:density} can be easily used to get the density of the Arithmetic Asian option under the Black-Scholes model. This would correspond to the case $\alpha=0$ and $\nu=\sigma>0$ as the Black-Scholes constant volatility.
\end{remark}

\begin{remark}
Assuming the density $\psi_{RV}$ exists, we have the following volatility swap price:
$$\mathbb{E}[\sqrt{RV(v^{(\gamma,\nu)})(T)}] =\int_0^\infty \sqrt{x}\psi_{RV}(x,T)\D x.$$
\end{remark}
In Figure \ref{fig:VolSwap}, we provide numerical results for the volatility swap approximation, which performs best for short maturities, due to the nature of the approximation being motivated by small-time smile behaviour.
Interestingly, it captures rather accurately the short time decay of the Volatility Swap price for maturities less than 3 months; for larger maturities the absolute error does not exceed 20 basis points.

\begin{figure}[h]
\centering
\includegraphics[scale=0.48]{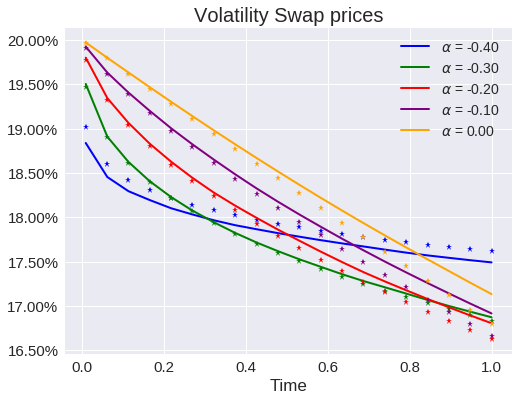}
\caption{Volatility Swap Monte Carlo price estimates (straight lines) and LDP based approximation (stars) for $\eta=1.5$ and $v_0=0.04$; for Monte Carlo we use $200,000$ simulations and $\Delta t =\frac{1}{1008}$.}
\label{fig:VolSwap}
\end{figure}

\clearpage
\section{Proof of Main Results}\label{app: proof of main results}
\subsection{Proof of Theorem \ref{LDP for vol process}} \label{3.3}
\begin{proof}
For $t\in\Tt$, $\eps>0$, we first define the rescaled processes
\begin{equation}\label{def: rescaled rBerg processes}
\begin{aligned}
Z^{\eps}_t &:=  \eps^{\beta/2}Z_t \overset{\textrm{d}}= Z_{\eps  t},  \\
v^{\eps}_t & :=   v_0 \exp \left( Z^{\eps}_t -\frac{\eta^2}{2} (\eps t)^{\beta} \right), &
\end{aligned}
\end{equation}
where $\beta := 2 \alpha + 1 \in (0,1)$.
From Schilder's Theorem \cite[Theorem  3.4.12 ]{DS89} and Proposition \ref{RKHS for v}, we have that the sequence of processes
$(Z^{\eps})_{\eps>0}$ satisfies a large deviations principle on $\Cc(\Tt)$ with speed $\eps^{-\beta}$ and rate function $\Lambda^Z$.
We now prove that the two sequences of stochastic processes $Z^\eps$ and
$\widetilde{Z}^\eps~:~=~ Z^\eps - \frac{\eta^2}{2}(\eps \cdot)^\beta$ are exponentially equivalent \cite[Definition 4.2.10]{DZ10}.
For each $\delta>0$ and $t \in \Tt$, there exists $\eps_* ~:~=~ \frac{1}{t} \left(\frac{2\delta}{\eta^2}\right)^{1/\beta}>0$ such that
$$ \sup_{t\in\Tt} \vert Z^\eps_t - \widetilde{Z}_t^\eps| = \sup_{t\in\Tt} \vert \frac{\eta^2}{2} (\eps t )^\beta \vert \le \delta,
$$
for all $0<\eps<\eps_*$.
Therefore, for all $\delta >0$, $\limsup_{\eps \downarrow 0} \eps^\beta \log \PP(\Vert Z^\eps - \widetilde{Z}^\eps\Vert_{\infty}>\delta) = - \infty$, and the two processes are indeed exponentially equivalent.
Then, using~\cite[Theorem 4.2.13]{DZ10}, the sequence of stochastic processes $(\widetilde{Z}^\eps)_{\eps >0 } $ also satisfies a large deviations principle on $\Cc(\Tt)$, with speed $\eps^{-\beta}$ and rate function $\Lambda^Z$.
Moreover, for all $\eps,t$, we have that $v^{\eps}_t = v_0 \exp(\widetilde{Z}^\eps_t)$, where the bijective transformation
$x~\mapsto~v_0\exp(x) $ is clearly continuous with respect to the sup norm metric.
Therefore we can apply the Contraction Principle \cite[Theorem 4.2.1]{DZ10}, which is stated in Appendix \ref{appendix: contraction princ}, concluding that the sequence of processes $(v^{\eps})_{\eps >0}$ satisfies a large deviations principle on $\Cc(\Tt)$ with speed
$\eps^{-\beta}$ and rate function $\Lambda^Z\left( \log\left(\frac{x}{v_0} \right)  \right) $.
Here we have used that, for each $\eps>0$, $t\in\Tt$ and $\xx \in \Cc(\Tt)$, the inverse mapping of the bijection transformation $\xx(t,\eps)\mapsto v_0\exp(x) $ is given by $\log\left(\frac{x}{v_0} \right)$.
Since, For all $\eps > 0 $ $\sup_{t \in \Tt} \left\vert v^\eps_t - v_{\eps t} \right\vert = 0 $, which implies that for any $\delta > 0$, 
$\limsup_{\eps \downarrow 0} \eps^{\beta} \log \mathbb{P} ( \Vert v^\eps_\cdot - v_{(\eps \cdot)} \Vert_{\infty} >   \delta) = -\infty$ trivially.
Thus the sequence of processes $(v^\eps_t)_{t \in \Tt}$ and $(v_{\eps t })_{t \in \Tt})$ are exponentially equivalent and therefore satisfy the same LDP.
Notice also that $\Lambda^v(v_0)=\Lambda^Z(0)=\Vert 0 \Vert^2_{\mathscr{H}^{K_\alpha} }=0.$
\end{proof}

\subsection{Proof of Corollary \ref{cor:ImpliedVolrBergomi}} \label{3.7}
\begin{proof}
The proof of Equation~\eqref{eq:OTM_IV} is similar to the proof of~\cite[Corollary 4.9]{FZ17}, and we shall prove  the lower and upper bound separately, which turn out to be equal.
Firstly, as the rate function $\hat{\Lambda}^v$ is continuous on $\Cc (\Tt)$, we have that, for all $k>0$,
$$
\lim_{t \downarrow 0} t^\beta \log \PP (\log [RV(v)(t)]>k) = -\mathrm{I}(k),
$$
as an application of Corollary~\ref{integrated vol LDP}.

\begin{itemize}
\item[(1)] The proof of the lower bound is exactly the same as presented in~\cite[Appendix C]{FZ17} and will be omitted here; we arrive at $\liminf_{t \downarrow 0} t^{\beta}\log \mathbb{E} \left[ (RV(v)(t)-\E^k)^+\right] \ge -\mathrm{I}(k). $
\item[(2)] To establish the upper bound, We closely follow \cite{PZ16} : \\
First we prove that 
\begin{equation}\label{eq:limequality}\lim_{t \downarrow 0} t^{\beta} \log \mathbb{E} [(RV(v)(t)-\E^k)]^+ = \lim_{t \downarrow 0} t^{\beta} \log \PP(RV(v)(t)\ge \E^k).  \end{equation}
We apply H\"older's inequality: 
$$
\mathbb{E} [(RV(v)(t)-\E^k)]^+ =
\mathbb{E} \left[ (RV(v)(t)-\E^k) \ind_{\{RV(v)(t)\ge \E^k\}} \right]
\le \mathbb{E} \left[ (RV(v)(t)-\E^k)^q \right]^{1/q} {\PP(RV(v)(t)\ge \E^k)}^{1-1/q},
$$
which holds for all $q>1$. We note that for all $q\geq 2$,  the mapping $x\mapsto x^q$ for any $x\geq 0$ is convex, therefore by H\"older's inequality one obtains $\left(x+y\right)^q\leq \frac{x^q +y^q}{2^{1-q}}$ for any $x,y \geq 0$, which in turn implies
\begin{equation}
\label{eq:3.7.1}
\mathbb{E} \left[ |RV(v)(t)-\E^k|^q \right]\leq 2^{q-1}\left[\mathbb{E} \left[ (RV(v)(t))^q\right]+(\E^k)^q \right].
\end{equation}
We further obtain the following inequality by applying Jensen's inequality and the fact that all moments exist for $(RV(v)(t))^q$
\begin{equation}\label{qth moment of RV}\mathbb{E} \left[ (RV(v)(t))^q\right] \leq \frac{1}{t^q} \int_0^t \mathbb{E}[v^q_s] \D s\leq  \frac{v^q_0}{t^q} \int_0^t \exp \left(\left(\frac{q^2\eta^2}{2} - \frac{q\eta^2}{2}\right) s^{2\alpha+1}\right) \D s\leq \frac{v^q_0}{t^{q-1}} \exp \left(\left(\frac{q^2\eta^2}{2} - \frac{q\eta^2}{2}\right)t^{2\alpha+1}\right)  \end{equation}
Therefore, using \eqref{eq:3.7.1} and \eqref{qth moment of RV} we obtain an upper bound
$$\limsup_{t \downarrow 0} t^{\beta} \log \mathbb{E} [(RV(v)(t)-\E^k)^+ ] \leq \limsup_{t \downarrow 0} (1-1/q) t^{\beta} \log \PP(RV(v)(t)\ge \E^k), $$
since it holds for $(1-1/q)< 1$.
To obtain a lower bound, we have for any $\varepsilon >0$ 
$$\mathbb{E} [(RV(v)(t)-\E^k)]^+\ge \mathbb{E} \left[(RV(v)(t)-\E^k)\ind_{\{RV(v)(t)\ge \E^k+\varepsilon\}}\right]\geq \varepsilon \PP\left(RV(v)(t)\ge \E^k+\varepsilon\right),$$
which implies
$$\liminf_{t \downarrow 0} t^{\beta} \log \mathbb{E} [(RV(v)(t)-\E^k)^+]  \geq  \liminf_{t \downarrow 0} t^{\beta} \log \PP\left(RV(v)(t)\ge \E^k+\varepsilon\right). $$
Now, using \eqref{eq:limequality} leads to $\limsup_{t \downarrow 0} t^\beta \log \mathbb{E} \left[ (RV(v)(t)-\E^k)^+\right] \le -\mathrm{I}(k)$.
The conclusion for Equation~\eqref{eq:OTM_IV} then follows directly.

\end{itemize}
The proof of Equation~\eqref{eq:OTM_sqIV} follows the same steps, after proving that the process $\sqrt{RV(v)}$ satisfies a large deviations principle on $\RR_+$. Indeed, as the function $x \mapsto x^2$ is a continuous bijection on $\RR_+$, we have that the square root of the integrated variance process $\sqrt{RV(v)}$ satisfies a large deviations principle on $\RR_+$ as $t$ tends to zero, with speed $t^{-\beta}$ and rate function $\hat{\Lambda}^v((\cdot)^2)$, using~\cite[Theorem 4.2.4]{DZ10}.
\end{proof}
\subsection{Proof of Theorem \ref{truncated var ldp theorem}}\label{3.10}
\begin{proof}
For brevity we set $n=2$, but for larger $n$, identical arguments can be applied.
 From  Schilder's Theorem \cite[Theorem  3.4.12 ]{DS89} and Proposition \ref{RKHS for v}, we have that the sequence of processes
$(Z^{\eps})_{\eps>0}$ satisfies a large deviations principle on $\Cc(\Tt)$ with speed $\eps^{-\beta}$ and rate function $\Lambda^Z$.
Define the operator $f:\Cc(\Tt)\rightarrow \Cc((\Tt),\RR^2)$ by $f(\xx):=(\frac{\nu_1}{\eta} \xx, \frac{\nu_2}{\eta}\xx)$, which is clearly continuous with respect to the sup-norm $\Vert \cdot \Vert_{\infty}$ on $\Cc(\Tt,\RR^2)$.
Applying the Contraction Principle then yields that the sequence of two-dimensional processes $((\frac{\nu_1}{\eta}Z^\eps, \frac{\nu_2}{\eta}Z^\eps))_{\eps>0}$ satisfies a large deviations principle on $\Cc(\Tt,\RR^2)$ as $\eps$ tends to zero with speed $\eps^{-\beta}$ and rate function
$$\tilde{\Lambda}(\yy,\zz):=\inf\{\Lambda^Z(\xx): f(\xx)=(\yy,\zz) \}=\inf\{\Lambda^Z(\frac{\eta}{\nu_1}\yy): \zz=\frac{\nu_2}{\nu_1}\yy \}. $$
Identical arguments to the proof of Theorem \ref{LDP for vol process} give that the sequences of processes
$((\frac{\nu_1}{\eta}Z^\eps, \frac{\nu_2}{\eta}Z^\eps))_{\eps>0}$
and
$((\frac{\nu_1}{\eta}Z^\eps-\frac{\nu_1^2}{2}(\eps\cdot)^\beta, \frac{\nu_2}{\eta}Z^\eps-\frac{\nu^2}{2}(\eps\cdot)^\beta ) )_{\eps>0}$
are exponentially equivalent, thus satisfy the same large deviations principle, with the same rate function and the same speed.

We now define the operator $g^\gamma:\Cc(\Tt,\RR^2)\rightarrow\Cc(\Tt)$ as
$g^\gamma(\xx,\yy)=v_0(\gamma e^\xx + (1-\gamma)e^\yy) $. For small perturbations $\delta^\xx, \delta^\yy \in \Cc(\Tt)$ we have that
$$ \sup_{t\in\Tt} \vert
g^\gamma(\xx+\delta^\xx, \yy+\delta^\yy)-g^\gamma(\xx,\yy)
\vert
\le
\vert v_0 \vert
\left(
\sup_{t\in\Tt} \vert
\gamma e^{\xx(t)}(e^{\delta^\xx(t)}-1)
\vert
+
\sup_{t\in\Tt} \vert
(1-\gamma) e^{\yy(t)}(e^{\delta^\yy(t)}-1)
\vert
\right).
$$
Clearly the right hand side tends to zero as $\delta^\xx,\delta^\yy$ tends to zero; thus the operator $g^\gamma$ is continuous with respect to the sup-norm $\Vert\cdot\Vert_{\infty}$ on $\Cc(\Tt)$.
Applying the Contraction Principle then yields that the sequence of processes
$ (v^{(\eps,\gamma,\nu)})_{\eps>0}:=
\left(v_0 \left(
\gamma\exp(\frac{\nu_1}{\eta}Z^\eps-\frac{\nu_1^2}{2}(\eps\cdot)^\beta)
+(1-\gamma)\exp(\frac{\nu_2}{\eta}Z^\eps-\frac{\nu_2^2}{2}(\eps\cdot)^\beta )
\right)\right)_{\eps>0}
$
satisfies a large deviations principle on $\Cc(\Tt)$ as $\eps$ tends to zero, with speed $\eps^{-\beta}$ and rate function
$$\xx\mapsto \inf\{\tilde{\Lambda}(\yy,\zz):\xx=g^\gamma(\yy,\zz)  \}
=\inf\{\Lambda^Z(\frac{\eta}{\nu_1}\yy):\xx=g^\gamma(\yy,\frac{\nu_2}{\nu_1}\yy)  \}
= \inf\{\Lambda^Z(\frac{\eta}{\nu_1}\yy):\xx=v_0 (\gamma e^\yy+(1-\gamma)e^{\frac{\nu_2}{\nu_1}\yy})  \}.
$$
Since, for all $\eps>0$ and $t\in\Tt$, $v^{(\gamma,\nu)}_{\eps t}$ and $v^{(\eps,\gamma,\nu)}_t$ are equal, the theorem follows immediately.
Identical arguments to the proof of Theorem \ref{LDP for vol process} then yield that
$\Lambda^\gamma(v_0)=0$.
\end{proof}
\subsection{Proof of Theorem~\ref{Th:LDP_mixedOU}}\label{3.14}
\begin{proof}
We begin by introducing a small-time rescaling of~\eqref{eq:mixedvarianceMultiFactorProcess} for $\eps >0$, so that the system becomes
\begin{equation}\label{eq:LDP_OU}
v^{(\gamma, \nu, \Sigma,\eps)}_t := v^{(\gamma, \nu, \Sigma)}_{\eps t}
= v_0 \sum_{i=1}^n \gamma_i \mathcal{E}\left( \frac{\nu^i}{\eta} \cdot  \mathrm{L}_i \mathcal{Z}^{\eps}_t \right),
\end{equation}
with the rescaled process $\mathcal{Z}^{\eps}_t$ defined as $\mathcal{Z}^{\eps}_t := \mathcal{Z}_{\eps t} = \eps^{\alpha+\frac{1}{2}} \left(\int_0^t K_\alpha(s,t)\D W^1_s, ..., \int_0^t K_\alpha(s,t)\D W^m_s \right)$.

The $m$-dimensional sequence of processes $(\eps^{\beta/2}(W^1,\cdots,W^m))_{\eps>0}$ satisfies a large deviations principle on $\mathcal{C}(\Tt,\RR^m)$ as $\eps$ goes to zero with speed $\eps^{-\beta}$ and rate function $\Lambda^m$ defined by $\Lambda^m(\mathrm{x}) := \frac{1}{2} \left\| \mathrm{x} \right\|^2_2$ for $\xx\in \mathscr{H}_m$ and $+\infty$ otherwise, by Schilder's Theorem \cite[Theorem  3.4.12 ]{DS89}.  
 $\mathscr{H}_m$ is the reproducing kernel Hilbert space of the measure induced by $(W_1,\cdots,W_m)$ on
 $\mathcal{C}(\Tt,\RR^m)$,
defined as
$$\mathscr{H}_m := \left\{ (g^1,\cdots, g^m)\in\mathcal{C}(\Tt,\RR^m): g^i (t) = \int_0^t f^i(s) \D s, f^i\in L^2(\Tt) \text{ for all } i\in 1\cdots m \right\}.$$
Then, using an extension of the proof of~\cite[Theorem 3.6]{HJL18}, for $i=1,\cdots, n$, the sequence of $m$-dimensional processes $\left(\mathrm{L}_i \mathcal{Z}^\eps_\cdot \right)_{\eps >0}$ satisfies a large deviations principle on $\mathcal{C}(\Tt,\RR^m)$ as $\eps$ tends to zero with speed $\eps^{-\beta}$ and rate function
$ \yy \mapsto \inf \left\{ \Lambda^m(\mathrm{x}): \mathrm{x}\in \mathscr{H}_m, \mathrm{y} = \mathrm{L}_i \mathrm{x}\right\}$ with $\mathrm{L}_i$ the lower triangular matrix introduced in Model~\eqref{eq:mixedvarianceMultiFactorProcess}.
Consequently, for $i=1,\cdots, n$ each (one-dimensional) sequence of processes
$\left( \frac{\nu^i}{\eta}\cdot\mathrm{L}_i \mathcal{Z}^\eps_\cdot \right)_{\eps >0}$ also satisfies a large deviations principle as $\eps$ tends to zero, with speed $\eps^{-\beta}$ and rate function
$ \Lambda_{\Sigma_i} (\mathrm{y}) := \inf \left\{ \Lambda^m(\mathrm{x}): \mathrm{x}\in \mathscr{H}_m, \mathrm{y} = \frac{\nu^i}{\eta}\cdot \mathrm{L}_i \mathrm{x}\right\}$.

Analogously to Theorem \ref{LDP for vol process}, each sequence of processes
$\left( \frac{\nu^i}{\eta} \cdot \mathrm{L}_i \mathcal{Z}^\eps_\cdot \right)_{\eps >0}$
and
$\left(\frac{\nu^i}{\eta}\cdot \mathrm{L}_i \mathcal{Z}^\eps_\cdot - \frac{1}{2} \nu^i \Sigma_i \nu^i (\eps\cdot)^\beta\right)_{\eps >0}
$
are exponential equivalent for $i=1\cdots n$ ; therefore they satisfy the same large deviations principle with the same speed $\eps^{-\beta}$ and the same rate function $\Lambda_{\Sigma_i}$.

We now define the operator $g^\gamma: \mathcal{C}(\Tt, \RR^n)  \rightarrow \mathcal{C}(\Tt)$
as
$$g^\gamma (\xx )(\cdot) := v_0 \sum_{i=1}^n \gamma_i \exp
\left( \frac{\nu^i}{\eta}\cdot \xx(\cdot) \right),$$
with $\xx :=(\xx_1,\cdots,\xx_n)$.
For small perturbations $\delta^1,\cdots, \delta^n \in \mathcal{C}(\Tt)$ with $\delta := (\delta^1,\cdots, \delta^n)$, we have that
\begin{align*}
\sup_{t\in \Tt} \left| g^\gamma (\xx+\delta)(t)-g^\gamma (\xx)(t) \right|
& =
\sup_{t\in \Tt} \left| v_0 \sum_{i=1}^n \gamma_i
\exp(\frac{\nu^i}{\eta}\cdot(\xx(t)+\delta(t)) ) -\exp(\frac{\nu^i}{\eta}\cdot\xx(t))
 \right| \\
& \le
\sup_{t\in \Tt}
\vert v_0 \vert
\sum_{i=1}^n
\left| \exp( \frac{\nu^i}{\eta}\cdot \xx(t)) (\exp(\delta(t))-1)  \right|
\end{align*}

The right-hand side tends to zero as $\delta^1,\cdots,\delta^n$ tends to zero; thus the operator $g^\gamma$ is continuous with respect to the sup-norm $\left\|\cdot \right\|_{\infty}$ on $\mathcal{C}(\Tt)$.
Using that
$v^{(\gamma, \nu, \Sigma,\eps)}_t= g^\gamma (\frac{\nu^i}{\eta}\cdot  \mathrm{L}_i \mathcal{Z}^\eps_\cdot - \frac{1}{2} \nu^i \Sigma_i \nu^i (\eps\cdot)^\beta)(t) $
for each $\eps >0$ and $t \in\Tt$,
we can apply  the Contraction Principle then yields that the sequence of processes $(v^{(\gamma, \nu, \Sigma,\eps)})_{\eps>0}$ satisfies a large deviations principle on $\mathcal{C}(\Tt)$ as $\eps$ tends to zero, with speed $\eps^{-\beta}$ and rate function
$$\mathrm{y}\mapsto \inf \left\{ \Lambda_{\Sigma_i} (\mathrm{x}) : \mathrm{y} = v_0 \sum_{i=1}^n \gamma_i \exp \left(\frac{\nu^i}{\eta}\cdot \mathrm{x} \right) \right\}
= \inf \left\{ \Lambda^m(\mathrm{x}) : \mathrm{x} \in \mathcal{H}_m, \mathrm{y} = v_0 \sum_{i=1}^n \gamma_i \exp \left( \frac{\nu^i}{\eta} \cdot \mathrm{L}_i \mathrm{x}\right) \right\}.$$
As with the previous two models we have that, for all $\eps>0$ and $t\in\Tt$,
$v^{(\gamma, \nu, \Sigma,\eps)}_t$ and $ v^{(\gamma, \nu, \Sigma)}_{\eps t}$ are equal in law and so the result follows directly.
\end{proof}
\subsection{Proof of Proposition \ref{rkhs of vix type process}}\label{5.2}
\begin{proof}
We reall that that the inner product
$ \langle \IgT f_1, \IgT f_2 \rangle_{\HgT}:= \langle f_1, f_2, \rangle_{L^2[0,T]}, $ where $$
\IgT f(\cdot)= \int_0^T g(\cdot,s) f(s) \D s, \quad
\HgT= \{\IgT f: f\in L^2[0,T] \}.$$The proof of Proposition \ref{rkhs of vix type process}, which is similar to the proofs given in \cite{JPS18}, is made up of three parts.
The first part is to prove that $ \left(\HgT,  \langle \cdot, \cdot \rangle_{\HgT} \right)$ is a separable Hilbert space.
Clearly $\IgT $ is surjective on $\HgT$.
Now take $f_1,f_2 \in L^2[0,T] $ such that $\IgT f_1 = \IgT f_2$. For any $t\in [T,T+\Delta]$ it follows that
$\int_0^T g(t,s)[f_1(s)-f_2(s)]\D s = 0 $; applying the Titchmarsh convolution theorem then implies that
$f_1=f_2$ almost everywhere and so $\IgT :L^2[0,T]\to \HgT$ is a bijection.
$\IgT $ is a linear operator, and therefore $ \langle \cdot, \cdot \rangle_{\HgT} $ is indeed an inner product; hence $ \left(\HgT,  \langle \cdot, \cdot \rangle_{\HgT} \right)$ is a real inner product space.
Since~$L^2[0,T]$ is a complete (Hilbert) space,
for every Cauchy sequence of functions $\{ f_n \}_{n \in \mathbb{N} }\in L^2[0,T] $,
there exists a function $\widetilde{f} \in  L^2[0,T]$
such that
$\{ f_n \}_{n \in \mathbb{N} }$ converges
to $\widetilde{f} \in  L^2[0,T]$; note also that
$\{\IgT f_n\}_{n\in\mathbb{N}}$ converges to $\IgT \widetilde{f}$ in $\HgT$.
Assume for a contradiction that there exists $f \in L^2[0,T]$ such that $f \neq \widetilde{f}$ and $\{\IgT f_n\}_{n\in\mathbb{N}}$ converges to $\IgT f$ in $\HgT$, then, since~$\IgT $ is a bijection,
the triangle inequality yields
$$
0 < \left\|\IgT f - \IgT \widetilde{f} \right\|_{\HgT}
\le\left\|\IgT f - \IgT f_n \right\|_{\HgT}
 +\left\|\IgT \widetilde{f} - \IgT f_n\right\|_{\HgT},
$$
which converges to zero as~$n$ tends to infinity.
Therefore $f=\widetilde{f}$, $\IgT f \in\HgT$ and $\HgT$ is complete,
hence a real Hilbert space.
Since~$L^2[0,T]$ is separable with countable orthonormal basis $\{\phi_n\}_{n \in \mathbb{N}}$,
then $\{\IgT  \phi_n\}_{n \in \mathbb{N}}$ is an orthonormal basis for~$\HgT$,
which is then separable.\\
The second part of the proof is to show that there exists a dense embedding
$\iota : \HgT \to \Cc[T,T+\Delta]$.
Since there exists $ h \in L^2[0,T]$ such that $\int_0^{\eps} |h(s)|\D s$ for all $\eps>0$  and
$g(t,\cdot)=h(t-\cdot)$ for any $t \in [T,T+\Delta]$, we can apply by \cite[Lemma 2.1]{Che08}, which tells us that
$\HgT$ is dense in $ \Cc[T,T+\Delta]$ and so we choose the embedding to be the inclusion map.\\
Finally we must prove that every $f^* \in \Cc[T,T+\Delta]^*$, where $f^*$ is defined in \cite[Definition 3.1]{JPS18}, is Gaussian on $ \Cc[T,T+\Delta]$, with variance
$\Vert \iota^* f^* \Vert_{{\HgT}^*} $, where $\iota^*$ is the dual of $\iota$.
Take $f^* \in \Cc[T,T+\Delta]^*$, then by the fact that $(\Cc[T,T+\Delta],\mathcal{H}^{g,T},\mu)$ is a RKHS triplet by similar arguments to \cite[Lemma 3.1]{FZ17}, we obtain using \cite[Remark 3.6]{JPS18} that $\iota^*$ admits an isometric embedding $ \overline{\iota}^*$ such that
$$ \Vert \overline{\iota}^* f^* \Vert_{{\HgT}^*}
=
\Vert f^* \Vert_{L^2(\Cc[T,T+\Delta] ,\mu)}
=
\int_{\Cc[T,T+\Delta]} \!\!\!\!\!\!\!\!\!\!\!\!\!\!\!\!\! (f^*)^2\D \mu
= \textrm{VAR}(f^*),
$$
where $\mu$ is the Gaussian measure induced by the  process $\int_0^T g(\cdot,s) \D s $ on
$\left( \Cc[T,T+\Delta], \mathscr{B}(\Cc[T,T+\Delta]) \right).$
\end{proof}
\subsection{Proof of Theorem \ref{them: rescaled VIX ldp}}\label{5.6}
\begin{proof}
First we recall $\widetilde{V}^{g,T,\eps }_t:= V^{g,T,\eps}_t - \frac{\eps^\gamma}{2}\int_0^tg^2(t,u)\D u + \eps^{\gamma /2}$.
We begin the proof by showing that the sequence of processes $(V^{g,T,\eps })_{\eps \in [0,1]}$ and
$(\widetilde{V}^{g,T,\eps})_{\eps \in [0,1]} $ are exponentially equivalent \cite[Definition 4.2.10]{DZ10}.
As $g(t,\cdot) \in L^2[0,T]$ for all $t\in[T,T+\Delta]$, for each $\delta>0$ there exists $\eps_*>0$ such that
$ \sup_{t\in[T,T+\Delta]} \left\vert \eps_*^{\gamma/2} - \frac{\eps_*^\gamma}{2}\int_0^T g^2(t,u)\D u \right\vert \le \delta.
$
Therefore, for the $\Cc[T,T+\Delta]$ norm $\Vert \cdot \Vert_{\infty}$ we have that for all $\eps_*>\eps>0$,
$$ \PP\left( \left\Vert V^{g,T,\eps } -\widetilde{V}^{g,T,\eps} \right\Vert_{\infty} > \delta \right) = \PP \left( \sup_{t\in[T,T+\Delta]} \left\vert \eps^{\gamma/2} - \frac{\eps^\gamma}{2}\int_0^T g^2(t,u)\D u \right\vert > \delta \right)=0.
$$

Therefore $\limsup_{\eps \downarrow 0} \eps^\gamma \log \PP\left( \left\Vert V^{g,T,\eps } -\widetilde{V}^{g,T,\eps} \right\Vert_{\infty} > \delta \right) = - \infty$, and so the two sequences of processes $(V^{g,T,\eps })_{\eps \in [0,1]}$ and
$(\widetilde{V}^{g,T,\eps})_{\eps \in [0,1]} $ are exponentially equivalent; applying \cite[Theorem 4.2.13]{DZ10} then yields that $(\widetilde{V}^{g,T,\eps})_{\eps \in [0,1]} $ satisfies a large deviations principle on $\Cc[T,T+\Delta]$ with speed $\eps^{-\gamma}$ and rate function $\Lambda^V$.

We now prove that the operators $\varphi_{1,\xi_0}$ and $\varphi_2$ are continuous with respect to the
$\Cc([T,T+\Delta]\times [0,1]) $ and $\Cc[0,1]$ $\Vert \cdot \Vert_{\infty} $ norms respectively. The proofs are very simple, and are included for completeness.
First let us take a small perturbation $\delta^f \in \Cc([T,T+\Delta]\times [0,1])$:
\begin{align*}
\left\Vert \varphi_{1,\xi_0}(f+\delta^f) - \varphi_{1,\xi_0}(f) \right\Vert_{\infty}
 & = \sup_{\substack{ \eps \in [0,1] \\ s \in [T,T+\Delta]}} \left\vert \xi_0(s)e^{f(s,\eps)}\left( e^{\delta^f(s,\eps)}-1\right)  \right\vert \\
 & \le \sup_{\substack{ \eps \in [0,1] \\ s \in [T,T+\Delta]}} \vert \xi_0(s) \vert
\sup_{\substack{ \eps \in [0,1] \\ s \in [T,T+\Delta]}}  \vert e^{f(s,\eps)} \vert
\sup_{\substack{ \eps \in [0,1] \\ s \in [T,T+\Delta]}}  \vert e^{\delta^f(s,\eps)}-1 \vert.
\end{align*}
Since $\xi_0$ is continuous on $[T,T+\Delta]$ and $f$ is continuous on $[T,T+\Delta]\times[0,1]$, they are both bounded.
Clearly $e^{\delta^f(s,\eps)}-1$ tends to zero as $\delta^f$ tends to zero and hence the operator $\varphi_{1,\xi_0}$ is continuous.
Now take a small perturbation $\delta^f \in \Cc([T,T+\Delta]\times [0,1])$:
$$ \left\Vert \varphi_2(f+\delta^f) -\varphi_2(f) \right\Vert_{\infty}
= \sup_{\eps\in[0,1]} \left\vert \frac{1}{\Delta} \int_T^{T+\Delta} \delta^f(s,\eps)\D s \right\vert
\le M,
$$
where $M:=\sup_{\eps\in[0,1]}\delta^f(s,\eps)$. Clearly $M$ tends to zero as $\delta^f$ tends to zero, thus the operator $\varphi_2$ is also continuous.

For every $s \in [T,T+\Delta]$ we have the following: by an application of the Contraction Principle \cite[Theorem 4.2.1]{DZ10} and using the fact that $\eps \mapsto (\varphi_{1,\xi_0}f)(s,\eps)$ is a bijection for all $f\in\Cc[T,T+\Delta]$
it follows that the sequence of stochastic processes
$
\left( \left(\varphi_{1,\xi_0} \widetilde{V}_s^{g,T,\eps}\right)(s,\eps)
\right)_{\eps\in[0,1]}
$
satisfies a large deviations principle on $\Cc[0,1]$ as $\eps$ tends to zero with speed $\eps^{-\gamma}$ and rate function
$$
\hat{\Lambda}^V_s(\yy):=\Lambda^V\left( (\varphi_{1,\xi_0} \yy)^{-1}(s,\cdot) \right)
=\Lambda^V\left(\log\left(\frac{\yy(s,\cdot)}{\xi_0(s)} \right)  \right).
$$
A second application of the Contraction Principle then yields that the sequence of stochastic processes
$\left( \
(\varphi_2 (\varphi_{1,\xi_0} \widetilde{V}_s^{g,T,\eps}))
(\eps)
\right)_{\eps\in[0,1]}
$
satisfies a large deviations principle on $\Cc[0,1]$ with speed $\eps^{-\gamma}$ and rate function
$\Lambda^{\text{VIX}}(\xx)=
\inf_{s\in[T,T+\Delta]}
\{ \Lambda^V\left( (\varphi_{1,\xi_0} \yy)^{-1}(s,\cdot) \right): \xx(\cdot) = (\varphi_2\yy)(\cdot)
\}.
$
By definition, the sequence of processes
$\left( \
(\varphi_2 (\varphi_{1,\xi_0} \widetilde{V}_s^{g,T,\eps}))
(\eps)
\right)_{\eps\in[0,1]}
$
is almost surely equal to the rescaled VIX processes
$(e^{\eps^{\gamma/2}}\text{VIX}_{T,\eps^{\gamma/2}})_{\eps \in [0,1]}
$ and hence the satisfies the same large deviations principle.

\end{proof}

\section{Numerical recipes}\label{appendix: numerics}
We first consider the simple rough Bergomi \eqref{rough Bergomi} model for sake of simplicity and further develop the mixed multi-factor rough Bergomi \eqref{eq:mixedvarianceMultiFactorProcess} model in Appendix \ref{App: MutiCase} (which also includes \eqref{eq:mixedvarianceProcess}).
Therefore, we tackle the numerical computation of the rate function $$\hat{\Lambda}^v(\yy) := \inf \left\{ \Lambda^v(\xx) : \yy = RV(\xx)(1) \right\}.$$
This problem, in turn,  is equivalent to the following optimisation:
\begin{equation}\label{eq:minimisation1}
\hat{\Lambda}^v(\yy):=\inf_{f\in L^2[0,1]} \left\{ \frac{1}{2}||f||^2 : \yy = RV\left(\exp\left(\int_0^\cdot K_{\alpha}(u,\cdot)f(u)\D u\right)\right)(1) \right\}.
\end{equation}
A natural approach is to consider a class of functions that is dense in $L^2[0,1]$. The Stone-Weierstrass theorem states that any continuous function on a closed interval can be uniformly approximated by a polynomial function. Consequently, we consider a polynomial basis,
$$\hat{f}^{(n)}(s)=\sum_{i=0}^n a_i s^i$$
such that $\{ \hat{f}^{(n)} \}_{a_i \in \RR}$ is dense in $L^2[0,1]$ as $n$ tends to $+\infty$. Problem \eqref{eq:minimisation1} may then be approximated via
$$\hat{\Lambda}^v_n(\yy):=\inf_{a\in\mathbb{R}^{n+1}} \left\{ \frac{1}{2}||\hat{f}^{(n)}||^2 : \yy = RV\left(\exp\left(\int_0^\cdot K_{\alpha}(u,\cdot)\hat{f}^{(n)}\D u(u)\right)\right)(1) \right\},$$
where $a=(a_0,...,a_n)$. In order to obtain the solution, first the constraint
$\yy = RV\left(\exp\left(\int_0^\cdot K_{\alpha}(u,\cdot)\hat{f}^{(n)}(u)\D u\right)\right)(1)$
needs to be satisfied. To accomplish this, we consider anchoring one of the coefficients in $\hat{f}^{(n)}$ such that

\begin{equation} \label{eq:coefanchoring}
a_i^*=\argmin_{a_i\in\mathbb{R}} \left\{\left(\yy-RV\left(\exp\left(\int_0^\cdot K_{\alpha}(u,\cdot)\hat{f}^{(n)}(u)\D u\right)\right)(1)\right)^2 \right\}
\end{equation}
and the constraint will be satisfied for all combinations of the vector $a^*=(a_0,...,a_{i-1},a_i^*,a_{i+1},...,a_n)$. Numerically, \eqref{eq:coefanchoring} is easily solved using a few iterations of the Newton-Raphson algorithm. Then we may easily solve
$$\inf_{a^*\in\mathbb{R}^{n+1}} \left\{ \frac{1}{2}||\hat{f}^{(n)}||^2\right\}$$
which will converge to the original problem \eqref{eq:minimisation1} as $n\to+\infty$. The polynomial basis is particularly convenient since we have that

\begin{equation}\label{eq:IVpolynomial}
RV\left(\exp\left(\int_0^\cdot K_{\alpha}(u,\cdot)\hat{f}^{(n)}(u)\D u\right)\right)(1)=\int_0^1 \exp\left(\eta \sqrt{2\alpha+1} \sum_{i=0}^n\frac{a_i s^{\alpha+1+i} \text{ }_2F_1(i+1,-\alpha,i+2,1)}{i+1}\right)\D s,
\end{equation}
where ${}_2F_1$ denotes the Gaussian hypergeometric function.
In particular one may store the values $\{  {}_2F_1 \left(i+1,-\alpha,i+2,1)\right)\}_{i=0}^n$ in the computer memory and reuse them through different iterations. In addition, the outer integral in \eqref{eq:IVpolynomial} is efficiently computed using Gauss-Legendre quadrature i.e.
$$RV\left(\left(\int_0^\cdot K_{\alpha}(u,\cdot)\hat{f}^{(n)}(u)\D u\right)\right)(1)\approx \frac{1}{2}\sum_{k=1}^m\exp\left(\eta\sqrt{2\alpha+1} \sum_{i=1}^n\frac{a_i \left(\frac{1}{2}(1+p_k)\right)^{\alpha+1+i}
\text{ }_2F_1(i+1,-\alpha,i+2,1)}{i+1}\right)w_k,$$
where $\{p_k,w_k\}_{k=1}^m$ are $m$-th order Legendre points and weights respectively.
\subsection{Convergence Analysis of the numerical scheme}
In Figure \ref{fig:ConvergenceAnalysis} we report the absolute differences $|\hat{\Lambda}^v_{n+1}(\yy)-\hat{\Lambda}^v_n(\yy)|$, for which we observe that the distance decreases as we increase $n$, the degree of the approximating polynomial. We must enphasize that both numerical routines performing the minimisation steps have a default tolerance of $10^{-8}$, hence we cannot expect to obtain higher accuracies than the tolerance and that is why Figure \ref{fig:ConvergenceAnalysis}  becomes noisy once this acuracy is obtained. We note that in Figure \ref{fig:ConvergenceAnalysis} we observe a fast convergence, and with just $n=5$  we usually obtain accuracy of $10^{-5}$.
\begin{figure}[h]
\centering
\includegraphics[scale=0.45]{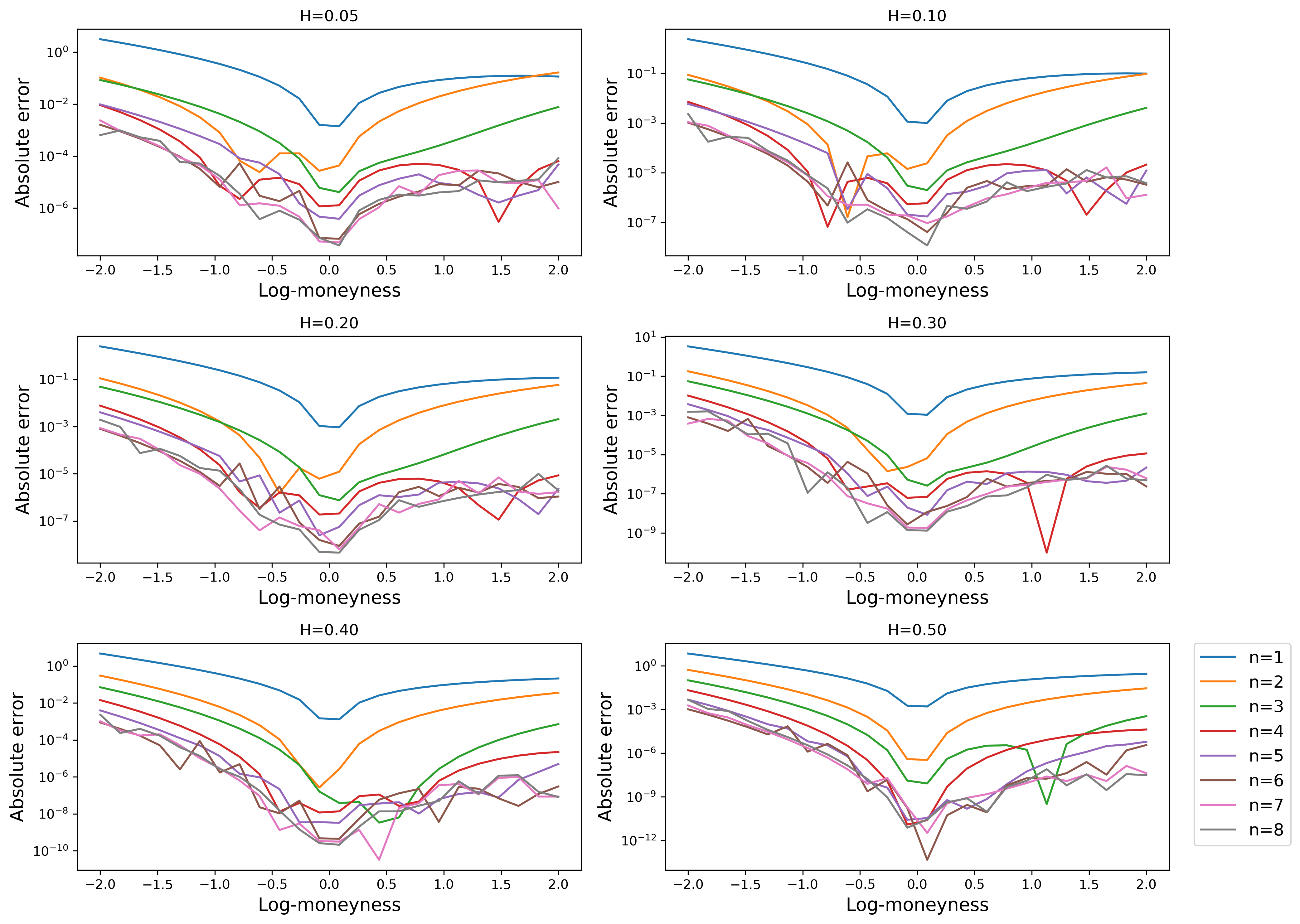}
\caption{Absolute difference of subsequent approximating polynomials $|\hat{\Lambda}^v_{n+1}(\yy)-\hat{\Lambda}^v_n(\yy)|$}
\label{fig:ConvergenceAnalysis}
\end{figure}
\subsection{A tailor-made polynomial basis for rough volatility}
We may improve the computation time of the previous approach by considering a tailor-made polynomial basis. In particular, recall the following relation
$$\int_0^s K_\alpha(u,s) u^k\D u=\frac{u^{\alpha+1+k} {}_2F_1(k+1,-\alpha,k+2)}{k+1},$$
then, for $k=-\alpha-1$ we obtain
$$\int_0^s K_\alpha(u,s) u^{-\alpha-1}\D u=\frac{ {}_2F_1(-\alpha,-\alpha,1-\alpha,1)}{-\alpha},$$
which in turn is a constant that does not depend on the upper integral bound $s$.
\begin{proposition}\label{prop:anchorCloseForm}
Consider the basis $\hat{g}^{(n)}(s)=c s^{-\alpha-1}+\sum_{i=0}^n a_i s^i$, where $c\in\mathbb{R}$. Then, for $c=c^*$ with
$$c^*=\frac{-\alpha}{\eta\sqrt{2\alpha +1} \text{}_2F_1(-\alpha,-\alpha,1-\alpha,1) }\log\left(\frac{y}{\int_0^1 \exp\left(\eta\sqrt{2\alpha + 1} (\sum_{i=0}^n\frac{a_i s^{\alpha+1+i} \text{ }_2F_1(i+1,-\alpha,i+2,1)}{i+1}\right) \D s}\right),$$
$\hat{g}^{(n)}(s)$ solves \eqref{eq:coefanchoring}.
\end{proposition}
\begin{proof}
We have that
\begin{align*}
RV\left(\int_0^\cdot K_{\alpha}(u,\cdot)\hat{g}^{(n)}(u)\D u\right)(1)&=\exp\left(\frac{\eta\sqrt{2\alpha+1}}{-\alpha}\text{ }_2F_1\left(-\alpha,-\alpha,1-\alpha,1\right)\right)\\
&\times\int_0^1 \exp\left(\eta\sqrt{2\alpha+1} \sum_{i=0}^n\frac{a_i s^{\alpha+1+i} \text{ }_2F_1\left(i+1,-\alpha,i+2,1\right)}{i+1}\right)\D s
\end{align*}
and the proof trivially follows by solving $y=RV\left(\int_0^\cdot K_{\alpha}(u,\cdot)\hat{g}^{(n)}(u)\D u\right)(1)$.
\end{proof}

\begin{remark}
Notice that Proposition \ref{prop:anchorCloseForm} gives a semi-closed form solution to \eqref{eq:coefanchoring}. Then, we only need to solve
$$\inf_{(a_0,...,a_n)\in\mathbb{R}^{n+1}} \left\{ \frac{1}{2}||\hat{g}^{(n)}||^2: c=c^*\right\}$$
in order to recover a solution for \eqref{eq:minimisation1}.
\end{remark}

\begin{remark}
Notice that $u^{-\alpha-1}\notin L^2[0,1]$, however $u^{-\alpha-1}\ind_{\{u>\eps\}}\in L^2[0,1]$ for all $\eps>0$. Moreover,
\begin{align*}\int_0^s K_\alpha(s,u) u^{-\alpha-1}\ind_{\{u>\eps\}}\D u&=\frac{\text{ }_2F_1(-\alpha,-\alpha,1-\alpha,1)}{-\alpha}-\frac{\varepsilon^{-\alpha}s^{\alpha}\text{ }_2F_1(-\alpha,-\alpha,1-\alpha,\frac{\varepsilon}{t})}{-\alpha}\\&=\frac{\text{ }_2F_1(-\alpha,-\alpha,1-\alpha,1)}{-\alpha}+\mathcal{O}(\eps^{-\alpha}),
\end{align*}
hence for $\eps$ sufficiently small the error is bounded as long as $\alpha\neq 0$. In our applications we find that this method behaves nicely for $\alpha \in(-0.5,-0.05]$. In Figure \ref{fig:TruncatedBasisErrors} we provide precise errors and we observe that the convergence is better for small $\alpha$ (which is rather surprising behaviour, as the converse is true of other approximation schemes when the volatility trajectories become more rough) as well as strikes around the money. Moreover, the truncated basis approach constitutes a 30-fold speed improvement in our numerical tests. As benchmark we consider the standard numerical algorithm introduced in  \eqref{eq:coefanchoring}, with accuracy measured by absolute error.
\end{remark}
\begin{figure}[h]
\centering
\includegraphics[scale=0.6]{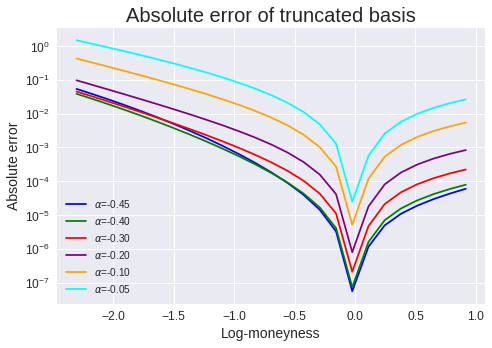}
\caption{Absolute error of the rate function. We consider the truncated basis approach against the standard polynomial basis with $\eta=1.5$, $v_0=0.04$ and different values of $\alpha$.}
\label{fig:TruncatedBasisErrors}
\end{figure}
\subsection{Multi-Factor case}\label{App: MutiCase}

The correlated mixed multi-factor rough Bergomi \eqref{eq:mixedvarianceMultiFactorProcess} model requires a slightly more complex setting. By Corollary \ref{integrated trunc var LDP} the rate function that we aim at is given by following multidimensional optimisation problem:
\begin{equation}\label{eq:minimisationMultiDim}
\hat{\Lambda}^{(v,\Sigma)}(\yy):=\inf_{(f_1,...,f_n)\in L^2[0,1]} \left\{ \frac{1}{2}\sum_{i=1}^n||f_i||^2 : \yy = RV\left(\sum_{i=1}^m \gamma_i \exp\left(\frac{\nu^i}{\eta}\cdot\Sigma_i \mathfrak{f}^{K_{\alpha}}_{.}\right)u\right)(1) \right\},
\end{equation}
where $\mathfrak{f}^{K_{\alpha}}_{.}=\left(\int_0^{\cdot} K_{\alpha}(u,\cdot)f_1(u)\D u,...,\int_0^{\cdot} K_{\alpha}(u,\cdot)f_n(u)\D u\right)$. The approach to solve this problem is similar to that of \eqref{eq:minimisation1}. Nevertheless, in order to solve \eqref{eq:minimisationMultiDim} we shall use a multi-dimensional polynomial basis
$$\left(\hat{f}^{(p)}_1(s),...,\hat{f}^{(p)}_n(s)\right)=\left(\sum_{i=0}^p a^1_i s^i,...,\sum_{i=0}^p a^n_i s^i\right)$$
such that each $\hat{f}^{(p)}_i(s)$ for $i\in\{1,...,n\}$ is dense as $p$ tends to $+\infty$ in  $L^2[0,1]$ by Stone-Weierstrass Theorem. Then we may equivalently solve
\begin{equation}\label{eq:minimisationMultiDim2}
\inf_{(a_0^1,...,a_p^1,...,a_0^n,...,a_p^n)\in \mathbb{R}^{(p+1)n}} \left\{ \frac{1}{2}\sum_{i=1}^n||\hat{f}^{(p)}_i||^2 : \yy = RV\left(\sum_{i=1}^m \gamma_i \exp\left(\frac{\nu^i}{\eta}\cdot\Sigma_i \hat{\mathfrak{f}}^{(K_{\alpha},p)}_{.}\right)u\right)(1) \right\},
\end{equation}
where $\hat{\mathfrak{f}}^{(K_{\alpha},p)}_{.}=\left(\int_0^{\cdot} K_{\alpha}(u,\cdot)\hat{f}^{(p)}_1(u)\D u,...,\int_0^{\cdot} K_{\alpha}(u,\cdot)\hat{f}^{(p)}_n(u)\D u\right)$. Then as $p$ tends to $+\infty$, \eqref{eq:minimisationMultiDim2} will converge to the original problem \eqref{eq:minimisationMultiDim}. In order to numerically accelerate the optimisation problem in \eqref{eq:minimisationMultiDim2}, we anchor coefficients $(a_0^1,....,a_0^n)$ to satisfy the constraint $y=RV(\cdot)(1)$ (same way we did in the one dimensional case), that is
$$\mathfrak{a}^*:=\inf_{(a_0^1,....,a_0^n)\in \mathbb{R}^n} \left\{\left( \yy - RV\left(\sum_{i=1}^m \gamma_i \exp\left(\frac{\nu^i}{\eta}\cdot\Sigma_i \mathfrak{f}^{K_{\alpha}}_{.}\right)u\right)(1)\right)^2 \right\}$$
where $\mathfrak{a}^*=(\mathfrak{a}^{1*}_0,...,\mathfrak{a}^{n*}_0)$ and one may use \eqref{eq:IVpolynomial} and Gauss-Legendre quadrature to efficiently compute $RV(\cdot)(1)$. Then, the constraint will always be satisfied by construction and instead we may solve
\begin{equation}\label{eq:minimisationMultiDim3}
\inf_{(\mathfrak{a}_0^{1*},a_1^1...,a_p^1,...,\mathfrak{a}_0^{n*},a_1^n,...,a_p^n)\in \mathbb{R}^{(p+1)n}} \left\{ \frac{1}{2}\sum_{i=1}^n||\hat{f}^{(p)}||^2 \right\}.
\end{equation}

\section{Exponential Equivalence and Contraction Principle}\label{appendix: contraction princ}

\begin{definition}\label{def:ExpEquiv}
On a metric space $(\mathcal{Y},d)$, two $\mathcal{Y}$-valued sequences ${(X^{\eps})}_{\eps > 0}$ and
${(\widetilde{X}^{\eps})}_{\eps > 0}$ are called exponentially equivalent
(with speed~$h_\eps$) if there exist probability spaces $(\Omega, \mathcal{B}_\eps, \PP_\eps)_{\eps>0}$ such that
for any $\eps>0$, $\PP^{\eps}$ is the joint law of $(\widetilde{X}^{\eps},X^{\varepsilon})$, $\varepsilon>0$ and,
for each $\delta >0$, the set $\left\{ \omega: (\widetilde{X}^{\eps},X^{\eps}) \in \Gamma_{\delta} \right\}$ is $\mathcal{B}_{\eps}$-measurable, and
\begin{equation*}
\limsup_{\eps \downarrow 0} h_\eps \log \PP^{\eps}\left(\Gamma_\delta\right)=- \infty,
\end{equation*}
where $\Gamma_\delta := \left\{ (\tilde{y},y): d(\tilde{y},y) > \delta \right\} \subset \mathcal{Y}\times\mathcal{Y}$.
\end{definition}

\begin{theorem}
Let $\mathcal{X}$ and $\mathcal{Y}$ be topological spaces and $f: \mathcal{X} \rightarrow \mathcal{Y}$ a continuous function. Consider a good rate function $I: \mathcal{X} \rightarrow [0,\infty]$. For each $y\in\mathcal{Y}$, define $I'(y) := \inf \{I(x): x\in\mathcal{X}, y=f(x)\}$. Then, if $I$ controls the LDP associated with a family of probability measures $\{\mu_\eps\}$ on $\mathcal{X}$, the $I'$ controls the LDP associated with the family of probability measures $\{\mu_\eps \circ f^{-1} \}$ on $\mathcal{Y}$ and $I'$ is a good rate function on $\mathcal{Y}$.
\end{theorem}

\section{Proof of $v_0\phi(d_1(x))=x \phi(d_2(x))$}\label{appendix: Proof_A1}
In order to prove
$v_0\phi(d_1(x))=x \phi(d_2(x))$,
we will prove the following equivalent result
$$
\left(d_1 (x) \right)^2 - \left(d_2 (x) \right)^2
= 2 \log \left(\frac{v_0}{x} \right).
$$

\begin{proof}
Recall that $\phi(\cdot)$ is the standard Gaussian probability density function. Using that for $x\ge 0$, $d_1(x)=\frac{\log(v_0)-\log(x)}{\hat{\sigma}(x,T)\sqrt{T}}+\frac{1}{2}\hat{\sigma}(x,T)\sqrt{T}$ and $d_2(x)=d_1(x)- \hat{\sigma}(x,T)\sqrt{T}$, we obtain
$$
\begin{aligned}\displaystyle
\left(d_1 (x) \right)^2 - \left(d_2 (x) \right)^2
&= \left(d_1 (x) \right)^2 - \left( d_1 (x) - \hat{\sigma}(x,T) \sqrt{T} \right)^2, \\
&= 2 d_1 (x) \hat{\sigma}(x,T) \sqrt{T} - T\hat{\sigma}^2 (x,T), \\
&= 2 \left[\log(v_0)-\log(x)+\frac{1}{2}\hat{\sigma}^2(x,T) T \right] - T\hat{\sigma}^2 (x,T), \\
&= 2 \log \left(\frac{v_0}{x} \right).
\end{aligned}
$$
\end{proof}

\end{document}